\documentclass[runningheads]{llncs}
\usepackage[T1]{fontenc}
\usepackage{graphicx}
\usepackage{xcolor}
\usepackage[T1]{fontenc}
\usepackage[utf8]{inputenc}
\usepackage{listings}
\usepackage{array}

\definecolor{hypcolor}{rgb}{0.85, 0.48, 0.0}
\lstdefinestyle{leanHH}{
  language=lean,
  moredelim=**[is][\color{hypcolor}]{|!}{!|},
}
\usepackage{amssymb}
\usepackage[misc]{ifsym}

\definecolor{keywordcolor}{rgb}{0.0, 0.0, 0.0}   
\definecolor{tacticcolor}{rgb}{0.0, 0.3, 0.8}    
\definecolor{commentcolor}{rgb}{0.4, 0.4, 0.4}   
\definecolor{symbolcolor}{rgb}{0.0, 0.0, 0.0}    
\definecolor{sortcolor}{rgb}{0.0, 0.0, 0.0}      
\definecolor{attributecolor}{rgb}{0.7, 0.1, 0.1} 

\usepackage{fancybox}
\makeatletter
\newenvironment{CenteredBox}{%
\begin{Sbox}}{
\end{Sbox}\centerline{\parbox{\wd\@Sbox}{\TheSbox}}}
\makeatother

\usepackage{amsmath}
\usepackage{algorithm2e}
\SetKwProg{Fn}{Function}{}{end}
\SetKwSwitch{Switch}{Case}{Other}{match}{}{case}{otherwise}{}{}
\SetKwFor{For}{for}{do}{}
\SetKwInOut{Input}{In}
\SetKwInOut{Output}{Out}
\SetKw{Continue}{continue}
\SetKw{Break}{break}
\usepackage{tikz}
\usetikzlibrary{calc,trees,positioning,arrows.meta,chains}
\usetikzlibrary{shapes.geometric,decorations.pathreplacing}
\usetikzlibrary{decorations.pathmorphing,shapes,matrix,shapes.symbols}
\usepackage{wrapfig}

\usepackage[firstpage]{draftwatermark}
\usepackage[hypertexnames = false]{hyperref}
\hypersetup{
  colorlinks=true,
  linkcolor=blue,
  urlcolor=blue,
  citecolor=blue
}
\usepackage{orcidlink}

\allowdisplaybreaks

\newcommand{\vEq}{\textcolor{hypcolor}{\ttfamily Eq}\xspace}
\newcommand{\vmap}{\textcolor{hypcolor}{\ttfamily map}\xspace}
\newcommand{\vrev}{\textcolor{hypcolor}{\ttfamily reverse}\xspace}
\newcommand{\vmaprev}{\textcolor{hypcolor}{\ttfamily map\_reverse}\xspace}
\newcommand{\vrevrev}{\textcolor{hypcolor}{\ttfamily reverse\_reverse}\xspace}
\newcommand{\vneggoal}{\textcolor{hypcolor}{\ttfamily neg\_goal}\xspace}

\newcommand{\vnewname}{\ensuremath{\mathit{newname}}\xspace}
\newcommand{\vlf}{\ensuremath{\mathit{lf}}\xspace}
\newcommand{\vlvar}{\ensuremath{\mathit{lvar}}\xspace}
\newcommand{\vargs}{\ensuremath{\mathit{args}}\xspace}
\newcommand{\vlargs}{\ensuremath{\mathit{largs}}\xspace}
\newcommand{\vatype}{\ensuremath{\mathit{atype}}\xspace}
\newcommand{\vaabst}{\ensuremath{\mathit{aabst}}\xspace}
\newcommand{\vbabst}{\ensuremath{\mathit{babst}}\xspace}
\newcommand{\vmatches}{\ensuremath{\mathit{matches}}\xspace}
\newcommand{\vmaxInsts}{\ensuremath{\mathit{maxInsts}}\xspace}
\newcommand{\vhi}{\ensuremath{\mathit{hi}}\xspace}
\newcommand{\vci}{\ensuremath{\mathit{ci}}\xspace}
\newcommand{\vactive}{\ensuremath{\mathit{active}}\xspace}
\newcommand{\vfront}{\ensuremath{\mathit{front}}\xspace}
\newcommand{\vtype}{\ensuremath{\mathit{type}}\xspace}
\newcommand{\vprevci}{\ensuremath{\mathit{prevci}}\xspace}
\newcommand{\vprevhi}{\ensuremath{\mathit{prevhi}}\xspace}
\newcommand{\vnewci}{\ensuremath{\mathit{newci}}\xspace}
\newcommand{\vnewhi}{\ensuremath{\mathit{newhi}}\xspace}
\newcommand{\vmonohi}{\ensuremath{\mathit{monohi}}\xspace}
\newcommand{\vnh}{\ensuremath{\mathit{nh}}\xspace}
\newcommand{\vnc}{\ensuremath{\mathit{nc}}\xspace}

\newcolumntype{K}[1]{>{\centering\arraybackslash}p{#1}}
\DeclareRobustCommand{\refappendix}[1]{%
  \ifdefined\cameraReady
    Appendix~\ref*{#1} of \cite{ThisPaperOnArxiv}%
  \else
    Appendix~\ref{#1}%
  \fi
}

\DeclareRobustCommand{\maybeappendix}[1]{%
  \ifdefined\cameraReady
    \refstepcounter{section}
  \else
    \section{#1}
  \fi
}

\begin{document}

\title{Lean-auto: An Interface between Lean 4 and Automated Theorem Provers}


  \author{
    Yicheng Qian\inst{1}\textsuperscript{(\Letter)}\orcidlink{0009-0008-0194-9572},
    Joshua Clune\inst{2}\orcidlink{0000-0003-4047-6196},
    Clark Barrett\inst{1}\orcidlink{0000-0002-9522-3084}, \\ and
    Jeremy Avigad\inst{2}\orcidlink{0000-0003-1275-315X} \\
  }

  \authorrunning{Y. Qian et al.}

  \institute{
    Stanford University, Stanford, USA \\
    \email{pratherc@stanford.edu}, \and
    Carnegie Mellon University, Pittsburgh, USA
  }

  \SetWatermarkAngle{0}
  \SetWatermarkText{\raisebox{13.4cm}{
    \hspace{0.1cm}
    \href{https://doi.org/10.5281/zenodo.15315259}{\includegraphics{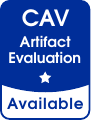}}
    \hspace{9cm}
    \includegraphics{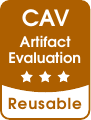}
  }}

\maketitle              
\begin{abstract}
  Proof automation is crucial to large-scale formal mathematics and software/hardware verification projects
  in ITPs. Sophisticated tools called hammers have been developed to provide general-purpose proof
  automation in ITPs such as Coq and Isabelle, leveraging the power of ATPs. An important
  component of a hammer is the translation algorithm from the ITP's logical system
  to the ATP's logical system. In this paper, we propose a novel translation algorithm
  for ITPs based on dependent type theory. The algorithm is implemented in Lean 4 under
  the name Lean-auto. When combined with ATPs, Lean-auto provides
  general-purpose, ATP-based proof automation in Lean 4 for the first time.
  Soundness of the main translation procedure is guaranteed, and experimental results
  suggest that our algorithm is sufficiently complete to automate the proof of many
  problems that arise in practical uses of Lean 4. We also find that Lean-auto
  solves more problems than existing tools on Lean 4's math library Mathlib4.
  
\end{abstract}

\section{Introduction}

  Interactive Theorem Provers (ITPs) \cite{Harrison2014HistoryOI}
  are widely used in formal mathematics and software/hardware verification.
  When using ITPs, straightforward but tedious proof
  tasks often arise during the proof development process.  Due to the limited
  built-in automation in ITPs, discharging these proof tasks can require significant
  manual effort.  Hammers
  \cite{Blanchette2016HammeringTQ,Czajka2018HammerFC} are proof automation tools for
  ITPs which utilize Automated Theorem Provers (ATPs, including Satisfiability
  Modulo Theories (SMT) solvers).
  Hammers have proved useful
  because they can solve many proof tasks automatically \cite{Paulson2012ThreeYO}.

  A hammer has three main components: premise selection, translation from ITP to
  ATP, and proof reconstruction from ATP to ITP. Premise selection collects
  the necessary premises (usually a list of theorems) needed to solve a proof task, translation exports
  the collected information from the ITP to the ATP, and proof reconstruction generates
  a proof in the ITP based on the output of the ATP.
  Our project Lean-auto primarily focuses on the translation from Lean 4 to ATPs.
  We note that Lean-auto does have a proof reconstruction procedure which fully supports
  one of the three types of ATPs we use to evaluate Lean-auto.
  For ATPs with proof reconstruction support, if the ATP successfully finds
  a proof, Lean-auto will generate proof terms and check them using the Lean 4 kernel.
  For other ATPs, if the ATP successfully finds a proof, Lean-auto will mark the
  problem as solved in Lean 4, but will generate a warning to indicate that Lean-auto
  trusts the ATPs' output. Ongoing projects are expected to implement premise selection and more
  proof reconstruction support. See Sect. \ref{sectexpr} for more discussion.
  
  The discrepancies between logical systems of ATPs and ITPs pose
  significant challenges to translation procedures between them.
  Several popular ITPs are based on highly expressive logical systems.
  For example, Isabelle \cite{Isabelle} is based on polymorphic higher-order logic, while
  Coq \cite{CoqRefMan} and Lean 4 \cite{Lean4}\footnote{Agda \cite{Agda}
  is also dependently typed, but is based on Martin-Löf type theory.} are based on an even more
  expressive logical system called dependent type theory (also known as
  $\lambda C$ in the lambda cube) \cite{LambdaWithType,Coquand1988}.\footnote{Or \textit{calculus of inductive
  constructions (CIC)}, depending on whether inductive types are considered as an extension.}
  Moreover, features such as typeclasses \cite{TypeClassHaskell}, universe polymorphism \cite{UPolyCoq}, and inductive types \cite{CICIndDef}
  are commonly used as extensions to the base logical system to enhance usability of the ITPs.
  On the other hand, ATPs are usually based on less expressive logical systems such
  as first-order logic (FOL) \cite{CVC5,Vampire,Z3Paper,EProver} and (in recent years)
  higher-order logic (HOL) \cite{HOVampire,ZipperpositionMakeWork,HOEProver}.
  An overview of the various logical systems relevant to our work is given in Sect. \ref{sublogsys}.

  \begin{wrapfigure}{L}{0.5\textwidth}
    \centering
    \tikzstyle{block} = [rectangle, draw, text width=12em, text centered, rounded corners, minimum height=3em]
    \begin{tikzpicture}
      [node distance=1.75cm, start chain=going below,]
      \node (n0) at (0,0) [text width=10em, text centered] {Lean 4};
      \node (n1) at ($(n0) + (0,-1.2cm)$) [block] {Preprocessing \linebreak (Sect. \ref{sectprep})};
      \node (n2) [block, below of=n1] {Quantifier instantiation \linebreak (Sect. \ref{sectinst})};
      \node (n3) [block, below of=n2] {$\lambda_\to^*$ abstraction \linebreak (Sect. \ref{sectabst})};
      \node (n4) [block, below of=n3] {Universe lifting \linebreak (\refappendix{appulift})};
      \node (n5) at ($(n4) + (0, -1.2cm)$) [text width=10em, text centered] {HOL};
      \draw [->] (n0) -- (n1);
      \draw [->] (n1) -- (n2) node[midway, right] {$\lambda C$};
      \draw [->] (n2) -- (n3) node[midway, right] {QMono $\lambda C$};
      \draw [->] (n3) -- (n4) node[midway, right] {$\text{HOL}^*$};
      \draw [->] (n4) -- (n5);
    \end{tikzpicture}
    \caption{Translation workflow of Lean-auto.}
    \label{fig:title}
  \end{wrapfigure}
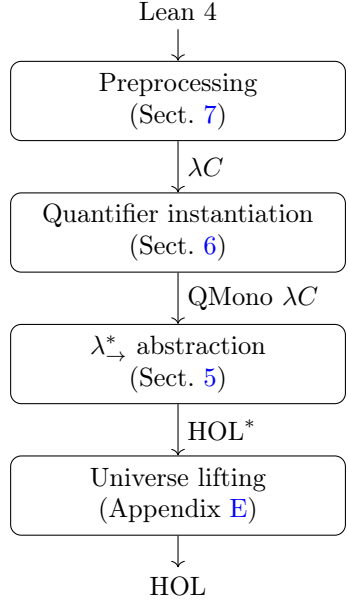

  There are two existing approaches for translation from more expressive
  logical systems to less expressive ones: encoding-based translation and monomorphization.
  Encoding-based translation is used in CoqHammer \cite{Czajka2018HammerFC}
  to translate Coq into untyped FOL. Monomorphization is used to
  eliminate polymorphism in Isabelle Sledgehammer \cite{Blanchette2016HammeringTQ,MonoPaper,Paulson2012ThreeYO}.
  Our small-scale experiment\footnote{See \refappendix{sstrans}.} on Mathlib4 suggests that encoding-based translation tends to produce
  much larger outputs than monomorphization, which could negatively affect the performance of ATPs.
  Therefore, we use monomorphization in Lean-auto. An overview of these two
  translation methods and related discussions are given in Sect. \ref{subencmon}.

  Since ATPs have started supporting HOL in recent years \cite{HOVampire,ZipperpositionMakeWork,HOEProver},
  Lean-auto translates Lean 4 to HOL. The overall translation has
  two stages: preprocessing and monomorphization.
  Monomorphization itself has three stages: quantifier instantiation,
  $\lambda_\to^*$ abstraction, and universe lifting.
  Roughly speaking, preprocessing translates Lean 4 into dependent type
  theory,\footnote{As mentioned before, Lean 4 is different from
  dependent type theory because it includes various additional language features.}
  and monomorphization translates dependent type theory into HOL.
  The monomorphization procedure of Lean-auto is inspired by Isabelle Sledgehammer.
  However, since dependent type theory is considerably different from
  Isabelle's HOL, the monomorphization procedure is thoroughly redesigned,
  and presented in a different way in our paper. Challanges related to
  dependent type theory and Lean 4 are discussed in Sect. \ref{exqdet}.

  In our paper, we work backwards in Lean-auto's translation workflow. We
  start from $\lambda_\to^*$ abstraction (Sect. \ref{sectabst}), then quantifier instantiation (Sect. \ref{sectinst}), and
  end with preprocessing (Sect. \ref{sectprep}). This is because it is easier to begin with the simpler
  logical system and progressively take into account more features of the
  highly expressive Lean 4 language. We leave universe lifting to \refappendix{appulift}
  since it is relatively straightforward compared to the other steps.

\subsection{Related Work}\label{sectrw}

  Hammers are not restricted to ITPs with expressive logical systems.
  Several ITPs based on FOL or HOL also have their hammers, for
  example, the hammer of Mizar \cite{Mizar40Paper}, the hammer of MetaMath \cite{MetamathHammer},
  and HOL(y)Hammer \cite{HolyHammerPaper}.
  Apart from hammers, there are various other ITP proof automation tools.
  For example, Coq and Lean both come with a \textit{tactics} language, and built-in tactics
  provide users with low-level proof automation, such as Coq's
  \texttt{apply}, \texttt{rewrite}, and \texttt{destruct} tactics~\cite{CoqRefMan}, and
  Lean's \texttt{apply}, \texttt{rw}, and \texttt{cases} tactics~\cite{ThmProvingLean4}.
  Domain-specific automation tools are also common, such as the intuitionistic
  propositional logic solver \texttt{tauto} of Coq, congruence closure algorithm \texttt{congruence}
  of Coq, and integer linear arithmetic solver \texttt{omega} of Lean 4, all implemented as tactics.
  Lightweight proof search procedures in ITPs include Coq's \texttt{auto} and Lean 4's Aesop \cite{AesopPaper}.
  There are also lightweight ATPs implemented in ITPs, such as
  Isabelle's Metis \cite{Hurd2003FirstOrderPT} and \texttt{blast\_tac} \cite{Paulson1999AGT}, HOL Light's Meson \cite{HOLMesonPaper},
  and Lean 4's Duper \cite{DuperPaper}. Finally, machine learning algorithms have
  also been used to automate proof in ITPs, for example, MagnusHammer \cite{mikula2023magnushammer}
  of Isabelle, LeanDojo \cite{Yang2023LeanDojoTP} of Lean,
  GPT-f \cite{GPTFPaper} of Metamath, and ASTactic \cite{Yang2019LearningTP} of Coq.

\section{Preliminaries}

\subsection{Dependent Type Theory}\label{subdtt}

  Dependent type theory, or $\lambda C$ in the $\lambda$-cube,
  or the \textit{calculus of constructions} (CoC) \cite{LambdaWithType},
  is a highly expressive type system and logical system. It is the logical
  foundation of Coq, Lean 4, and Agda. To align with
  Lean 4, we use the variant of $\lambda C$ which contains a countable
  number of non-cumulative universe levels. The syntax of $\lambda C$ terms is defined
  inductively as follows:
  $$\mathcal{T}_C ::= V \ | \ \mathsf{U}_\ell \ | \ \mathcal{T}_C \ \mathcal{T}_C \ |
    \ \lambda (V : \mathcal{T}_C). \mathcal{T}_C \ | \ \forall (V : \mathcal{T}_C). \mathcal{T}_C,$$
  where $V$ is the set of variables, $\mathsf{U}_\ell\ (\ell \in \mathbb{N})$ are
  the sorts (i.e., the types of types), $\mathcal{T}_C \ \mathcal{T}_C$ is function application,
  $\lambda (V : \mathcal{T}_C). \mathcal{T}_C$ is $\lambda$ abstraction, and
  $\forall (V : \mathcal{T}_C). \mathcal{T}_C$ is the product type.
  $\ell$ is called the universe level of $\mathsf{U}_\ell$.
  We use $\forall$ instead of $\mathrm{\Pi}$ to align with the syntax of Lean, Coq, and Agda.
  Syntactical equality of terms will be denoted as $=$, and $\beta\eta$-equivalence of terms will be
  denoted as $\cong$.
  
  We adopt the following commonly-used notational conventions:
  function application binds stronger than $\lambda$ and $\forall$, and is left-associative;
  consecutive $\lambda$s and $\forall$s can be merged, and $\lambda$s and $\forall$s with the same
  binder type can be further merged into the same parenthesis; when the product type is non-dependent,
  $\to$ can be used instead of $\forall$. Importantly, $\to$ binds stronger than $\forall$, i.e.,
  $\forall (x : \alpha). \beta \to \gamma$ is interpreted as $\forall (x : \alpha). (\beta \to \gamma)$
  instead of $(\forall (x : \alpha). \beta) \to \gamma$, the latter being the convention in
  FOL and HOL. The abbreviations $\bot, \neg, \land, \lor, \leftrightarrow,
  =_\ell,$ and $\exists_\ell$
  are defined in the usual way.\footnote{See \refappendix{applsc}.}

  A context $\Gamma$ is a list of variable declarations $x_1 : \alpha_1, \dots, x_n : \alpha_n$.
  Type judgements will be written as $\Gamma \vdash t : \alpha$, which stands for ``$\lambda C$ term $t$
  has type $\alpha$ under context $\Gamma$.''\footnote{Derivation rules for type judgements of
  $\lambda C$ are given in \refappendix{apppts}.} If $\Gamma \vdash t : \alpha$, then $t$ is
  called a \textit{well-formed term}, and $\alpha$ is called a (well-formed)
  \textit{type}.\footnote{In $\lambda C$, all well-formed types are also well-formed terms.}
  Under context $\Gamma$, a type $\alpha$ is called \textit{inhabited} iff there exists $t$ such
  that $\Gamma \vdash t : \alpha$, in which case $t$ is called an \textit{inhabitant} of $\alpha$.
  Propositions are types of type $\mathsf{U}_0$. A \textit{proof} of a proposition
  $p : \mathsf{U}_0$ is an inhabitant of $p$. A proposition $p : \mathsf{U}_0$ is \textit{provable}
  iff it is inhabited. Given a context $\Gamma$ and a proposition $p$, we use $\Gamma \vdash? p$ to
  represent the \emph{problem} of finding a proof of $p$ under context $\Gamma$.

  For a function $f : \forall (x_1 : \alpha_1) \ \dots \ (x_n : \alpha_n). \beta$ (here $\beta$ may begin with $\forall$),
  the $n$th argument of $f$ is called a \textit{static dependent argument} iff $x_n$ occurs in $\beta$.
  In many cases, static dependent arguments are also type arguments; for example, the first and second arguments
  of $\mathsf{List.map} : \forall (\alpha \ \beta : \mathsf{U}_1). (\alpha \to
  \beta) \to \mathsf{List} \ \alpha \to \mathsf{List} \ \beta$ are both static
  dependent arguments.
  Another important concept is \textit{dependent argument}.\footnote{See \refappendix{labstalgo}
  for its formal definition.} In practical scenarios, ``dependent argument'' and ``static dependent argument'' usually
  have the same meaning. Their intricate difference  is explained in Sect. \ref{exqdet}.

  We use $\lambda C$ notation for all logical systems that can be embedded in $\lambda C$.
  When presenting Lean 4 examples, we use additional Lean 4 notational conventions.
  These are explained in Sect. \ref{sectlean}.

\subsection{Logical Systems of ITPs and ATPs}\label{sublogsys}

  In this section, we give an overview of the various logical systems
  that are relevant to our work. In the following list, the logical systems
  are ordered from the least expressive to the most expressive. Note that, except
  for $\lambda C$ and more expressive systems, all other logical systems have
  two components: term calculus (which specifies the construction and computation
  rules of terms), and logical axioms/rules.
  \begin{enumerate}
    \item Untyped FOL, or predicate logic.
    \item Many-sorted FOL.
    \item Many-sorted HOL (monomorphic HOL, or just HOL), where functions are allowed
      to take functions as arguments, and quantifiers can quantify over functions.
      Its term calculus is simply typed lambda calculus $\lambda_\to$ \cite{LambdaWithType}.
    \item Many-sorted HOL with a countable number of universe levels, denoted as $\text{HOL}^*$,
      which is discussed in Sect \ref{sectpts}. This is an
      intermediate logical system used in Lean-auto's monomorphization.
      It is essentially equivalent to HOL\footnote{In \refappendix{appulift}, we
      show that $\text{HOL}^*$ is essentially equivalent to HOL.}.
    \item HOL with rank-1 polymorphism, or polymorphic HOL. Its term calculus is
      $\lambda 2$ in the $\lambda$-cube \cite{LambdaWithType}. In polymorphic HOL, functions are
      allowed to take type arguments, and quantifiers can quantify over types. However, type
      constructors, or types dependent on types, are not allowed.
    \item Isabelle's logical system. Based on polymorphic HOL. Supports (co)inductive datatypes and recursive functions.
    \item Dependent type theory, or $\lambda C$. Compared to polymorphic HOL, types can depend on terms and types in $\lambda C$.
    \item Coq, Lean 4, and Agda's logical systems. Based on $\lambda C$. Extensions to $\lambda C$ that are
      present in (at least one of) these ITPs include (co)inductive types, universe levels,
      universe polymorphism, typeclasses, and many others.
  \end{enumerate}
  
  \noindent All previously mentioned hammers translate between these logical systems. Isabelle
  Sledgehammer translates between Isabelle and
  HOL/FOL.\footnote{The exact logical system depends on the mode being used.}
  CoqHammer translates between Coq and untyped FOL. Lean-auto translates
  between Lean 4 and monomorphic HOL. As mentioned before, Lean-auto's
  preprocessing translates Lean 4 into $\lambda C$, and monomorphization
  translates $\lambda C$ into HOL. More specifically, quantifier instantiation
  and $\lambda_\to^*$ abstraction translates $\lambda C$ into $\text{HOL}^*$,
  and universe lifting translates $\text{HOL}^*$ into HOL.

\subsection{Pure Type Systems $\lambda C, \lambda_\to, \lambda_\to^*$ and Related Logical Systems}\label{sectpts}

  The Pure Type System (PTS) \cite{LambdaWithType} formalism enables concise specification
  of a class of type systems. We use PTS to formally specify the underlying type systems
  of the logical systems used in Lean-auto's translation.

  The specification of a PTS consists of a triple $(\mathcal{S}, \mathcal{A}, \mathcal{R})$,
  where $\mathcal{S}$ is the set of \textit{sorts}, $\mathcal{A} \subseteq \mathcal{S} \times \mathcal{S}$ is
  the set of \textit{axioms}, and $\mathcal{R} \subseteq \mathcal{S} \times \mathcal{S} \times \mathcal{S}$
  is the set of \textit{rules}. An axiom $(s_1, s_2) \in \mathcal{A}$ is intended to represent
  the typing axiom $s_1 : s_2$. The syntax of PTS terms is given by
  $$\mathcal{T} ::= V \ | \ \mathcal{S} \ | \ \mathcal{T} \ \mathcal{T} \ |
    \ \lambda (V : \mathcal{T}). \mathcal{T} \ | \ \forall (V : \mathcal{T}). \mathcal{T}$$
  
  Three type systems, $\lambda C$, $\lambda_\to$, and $\lambda_\to^*$, will be formulated
  using PTS.\footnote{The derivation rules of PTS
  are given in \refappendix{apppts}.} As mentioned above, $\lambda_\to$
  is the term calculus of HOL, and $\lambda_\to^*$ is the term
  calculus of $\text{HOL}^*$. Note that $\mathsf{U}_0$ is not present in $\lambda_\to$ and $\lambda_\to^*$
  because it is a special sort for propositions in $\lambda C$.
  The type of propositions in $\text{HOL}$ and $\text{HOL}^*$ will be represented by
  a special symbol $\mathsf{Bool} : \mathsf{U}_1$.
  
  $\lambda_\to^*$ and $\lambda_\to$ are similar, except that $\lambda_\to^*$ allows a
  countable number of universe levels $\ell \in \mathbb{N}^*$, where $\mathbb{N}^*$ is the set of
  positive integers. For example,
  in the type $(\alpha \to \beta) \to \gamma$, the subterms $\alpha, \beta,$ and $\gamma$ must be of type $\mathsf{U}_1$
  in the system $\lambda_\to$; however, in $\lambda_\to^*$, it is possible that $\alpha : \mathsf{U}_{\ell_1},
  \beta : \mathsf{U}_{\ell_2}, \gamma : \mathsf{U}_{\ell_3}$ where $\ell_1, \ell_2, \ell_3$
  may be different. A technicality related to PTS requires the presence of
  the sorts $\mathsf{U}_\ell'$ in $\lambda_\to^*$, with axioms $\mathsf{U}_\ell : \mathsf{U}_\ell'$.

  The logical systems HOL and $\text{HOL}^*$ are $\lambda_\to$ and $\lambda_\to^*$ augmented
  with the symbols $\mathsf{Bool}, \bot', \to', \forall'_s(\text{for each type }s)$, their
  corresponding typing rules, and logical rules. The abbreviations $\land', \lor', \neg', \leftrightarrow, ='_s, \exists'_s$
  are defined in a way consistent with their $\lambda C$ counterparts.
  The set of HOL and $\text{HOL}^*$ terms are denoted as $\mathcal{T}_\to$ and $\mathcal{T}_\to^*$, respectively.\footnote{
  The specifications of $\lambda C, \lambda_\to$, and $\lambda_\to^*$ using PTS are given in \refappendix{applll}.
  The formal definitions of $\text{HOL}$ and $\text{HOL}^*$ are given in \refappendix{apphol}.}

\subsection{Lean and Mathlib}\label{sectlean}

  Lean is an ITP based on dependent type theory. Lean-auto
  is implemented in Lean 4, the latest version of Lean. At present, the
  most prominent project in Lean is Mathlib \cite{MathlibPaper},
  which was renamed to Mathlib4\footnote{GitHub link: https://github.com/leanprover-community/mathlib4}
  when it was moved to Lean 4. Notably, Mathlib is the foundation of the
  Liquid Tensor Experiment \cite{LiquidTensor}, which successfully
  formalizes cutting-edge results in mathematics. 

  We will follow Lean 4 conventions when presenting Lean 4 examples. \texttt{Sort $\ell$}
  represents $\mathsf{U}_\ell$, and \texttt{Type $\ell$} represents $\mathsf{U}_{\ell + 1}$.
  \texttt{Sort $1$} (or \texttt{Type $0$}) can be abbreviated as \texttt{Type},
  and \texttt{Sort $0$} can be abbreviated as \texttt{Prop}. All user-declared
  symbols, including functions, are called \textit{constants} in Lean 4.
  Constants can have universe level parameters, but for simplicity,
  they are not shown in many of our Lean 4 examples. 
  Functions are allowed to have implicit arguments, which are represented by
  $\{x : \alpha\}$ instead of $(x : \alpha)$ in the type of the function.
  Prepending \textrm{@} to the name of a function causes implicit arguments
  to become explicit. For example, given the polymorphic list map function
  with the first and second argument being implicit:

  \centerline{\texttt{List.map : $\forall$ \{$\alpha \ \beta$ : Type\}, ($\alpha$ → $\beta$) → List $\alpha$ → List $\beta$},}
  
  \noindent the expression \texttt{@List.map $\alpha$ $\beta$ f} is the same as \texttt{List.map f}, where \texttt{f : $\alpha$ → $\beta$}.

  Typeclasses are extensively used by Lean 4's built-in library and Mathlib4 to
  overload arithmetic operators and represent mathematical structures. For example,
  consider the \texttt{HAdd} typeclass and the \texttt{HAdd.hAdd} function
  used to represent the addition operator in Lean 4.
  
  \centerline{\texttt{HAdd : $\forall$ ($\alpha \ \beta \ \gamma$ : Type), Type}}
  \centerline{\texttt{HAdd.hAdd : $\forall$ \{$\alpha \ \beta \ \gamma$ : Type\} [self : HAdd $\alpha \ \beta \ \gamma$], $\alpha$ → $\beta$ → $\gamma$}}

  \noindent An inhabitant of \texttt{HAdd $\alpha \ \beta \ \gamma$}, called a
  typeclass instance, is a wrapper of a ``heterogeneous'' addition operator, with $\alpha$ and $\beta$ as its input types and
  $\gamma$ as its output type. The square bracket in the type of \texttt{HAdd.hAdd}
  indicates that the enclosed argument is an instance argument, which is a special type
  of implicit argument intended to be filled by Lean 4's typeclass inference algorithm.
  Given the syntax \texttt{x + y} where \texttt{x : $\alpha$} and \texttt{y : $\beta$},
  the typeclass inference algorithm will attempt to find a type $\gamma$ and
  an instance \texttt{inst : HAdd $\alpha \ \beta \ \gamma$}, and elaborate the
  syntax \texttt{x + y} into the expression \texttt{@HAdd.hAdd $\alpha \ \beta \ \gamma$ inst x y}.
  In \texttt{@HAdd.hAdd $\alpha \ \beta \ \gamma$ inst}, the \texttt{HAdd.hAdd} function
  unwraps \texttt{inst} and returns the addition operator. This provides a mechanism
  for overloading operators. The same mechanism is used to represent mathematical
  structures in Mathlib4.

  Lean 4 supports definitional equality. Two terms are definitionally equal
  iff they can be converted to each other via Lean 4's built-in conversion rules.
  To test definitional equality of two terms $s$ and $t$, we can either reduce $s$ and
  $t$ to their normal forms and check syntactical equality, or use the optimized
  built-in function \texttt{isDefEq}\footnote{Its full Lean 4 name is \texttt{Lean.Meta.isDefEq}.}
  which checks definitional equality of a pair of terms.

  Inductive type is another important Lean 4 feature relevant to Lean-auto.
  It is handled by Lean-auto's preprocessing stage and is discussed in Sect. \ref{sectprep}.

  Lean 4 supports classical axioms such as function extensionality, excluded middle,
  and axiom of choice. Plain CoC does not include classical axioms. In contrast, classical axioms
  are built-in\footnote{They are either declared as axioms or derived from previously declared axioms, and
  they are imported during initialization.} in Lean 4. Lean-auto uses them during proof reconstruction.

\section{Encoding-based Translation and Monomorphization}\label{subencmon}

  Encoding-based translation and monomorphization
  are two approaches to translating from more expressive logical systems to
  less expressive logical systems.

  The idea behind encoding-based translations is to encode
  constructions in the more expressive system using function symbols in the less
  expressive system and to define the translation as a recursive function on the terms and formulas
  of the more expressive system. For example, in the dependent type theory of Coq,
  we have the type judgement relation $\Gamma \vdash x : w$, which means ``$x$ is of
  type $w$ under context $\Gamma$.'' There is no direct equivalent of this
  typing relation in untyped FOL. To express Coq type judgements in untyped FOL, 
  CoqHammer first introduces the uninterpreted FOL predicate $T(u^*, a^*)$, where
  $u^*$ and $a^*$ are FOL terms translated from Coq term $u$ and \textit{atomic} Coq type $a$
  (here \textit{atomic} roughly means that $a$ cannot be
  further decomposed by the translation procedure of CoqHammer). Then, a recursive function
  $\mathcal{G}_\Gamma(u, w)$ is defined on the Coq context $\Gamma$ and the Coq terms $u, w$.
  The function $\mathcal{G}_\Gamma(u, w)$ translates the typing relation $\Gamma \vdash u : w$ into an untyped FOL formula,
  in which the $T$ predicate is used to express type judgements involving atomic types.

  Encoding-based translation has the advantage of being (almost) complete
  and straightforward to compute. However, certain features of the more expressive
  logical system must be omitted to produce translation results of reasonable size,
  which sacrifices soundness \cite{Czajka2018HammerFC}. Moreover, even with this tradeoff,
  the translated expression is usually much larger than the original expression.

  The idea behind monomorphization is the fact that the proof of many propositions in the more expressive
  system can \textit{essentially} be conducted in the less expressive system. For example,
  in polymorphic HOL, given
  \begin{enumerate}
    \item the list map function $\mathsf{List.map} : \forall (\alpha \ \beta : \mathsf{U}_1). (\alpha \to \beta) \to \mathsf{List} \ \alpha \to \mathsf{List} \ \beta$
    \item two lists of natural numbers $xs \ ys : \mathsf{List} \ \mathbb{N}$ and two functions $f \ g : \mathbb{N} \to \mathbb{N}$
    \item the premise $xs = ys \land f = g$
  \end{enumerate}
  The equality
  \begin{equation}\label{lmapphol}
    \mathsf{List.map} \ \mathbb{N} \ \mathbb{N} \ f \ xs = \mathsf{List.map} \ \mathbb{N} \ \mathbb{N} \ g \ ys
  \end{equation}
  is provable using two rewrites $xs \Rightarrow ys, f \Rightarrow g$. The crucial observation is that, although $\textsf{List.map}$ is polymorphic, the term
  $\mathsf{List.map} \ \mathbb{N} \ \mathbb{N}$ as a whole behaves just like a monomorphic function,
  and therefore the rewrites can essentially be performed in monomorphic HOL. More formally,
  the formula \eqref{lmapphol} is the image of the monomorphic HOL formula $h \ f^* \ xs^* = h \ g^* \ ys^*$
  under the inter-logical-system ``substitution''
  $$\sigma := \{h \mapsto \mathsf{List.map} \ \mathbb{N} \ \mathbb{N},
    f^* \mapsto f, g^* \mapsto g, xs^* \mapsto xs, ys^* \mapsto ys\},$$
  and the rewrites $xs \Rightarrow ys, f \Rightarrow g$ in polymorphic HOL are just manifestations of the
  rewrites $xs^* \Rightarrow ys^*, f^* \Rightarrow g^*$ in monomorphic HOL.

  Monomorphization is sound, produces small translation results, and preserves
  term structures during translation. However, monomorphization is incomplete,
  since it is not always possible to find an appropriate formula in the less
  expressive logical system that reflects the original formula in the more expressive logical system.


  The difference in output size between encoding-based translation
  and mono\-morph\-ization is particularly pronounced in Lean 4 (see
  \refappendix{sstrans} for experimental results). As mentioned in Sect. \ref{sectlean},
  a user-facing Lean 4 syntax as simple as $x + y$ corresponds to
  the complicated expression \texttt{HAdd.hAdd $\alpha \ \beta \ \gamma$ inst x y}, where \texttt{inst}
  itself is a potentially large expression synthesized by typeclass inference. The result
  of encoding-based translation on the above expression is larger than the expression itself.
  On the other hand, our monomorphization procedure will translate the above expression
  into a much smaller one: $h \ x^* \ y^*$, where \texttt{HAdd.hAdd $\alpha \ \beta \ \gamma$ inst}
  is ``absorbed'' into $h$ via the inter-logical-system ``substitution.''

\section{An Overview of Lean-auto} \label{motex}

As mentioned before, the translation workflow of Lean-auto consists of four stages:
preprocessing, and the three stages of monomorphization: quantifier instantiation, $\lambda_\to^*$ abstraction,
and universe lifting.

Roughly speaking, the preprocessing stage translates Lean 4 into dependent type theory ($\lambda C$), which involves
handling definitional equality and inductive types. It also performs minimal
transformation on the translated $\lambda C$ problem. This includes introducing all leading
$\forall$ quantifiers into the context and applying proof by
contradiction.\footnote{Proof by contradiction introduces the negation of the
goal into the context and replaces the goal with $\bot$.} Then, everything in the context
with type \texttt{Prop} is collected by Lean-auto and added to
the list of premises. Sect.~\ref{sectprep} contains a more detailed discussion of preprocessing.

Universe lifting translates $\text{HOL}^*$ into $\text{HOL}$. Conceptually, it
erases all the universe level information in the input expression. However,
implementing it as a sound translation procedure in Lean 4 requires a decent amount
of work. Details about universe lifting are given in \refappendix{appulift}.

In Sect.~\ref{exabst} and~\ref{exinst}, we provide intuition for the
$\lambda_\to$ abstraction and quantifier instantiation stages by giving a
simplified explanation of their execution on an example. Sect.~\ref{exqdet}
gives a high-level discussion of some of the challenges posed by dependent
type theory and Lean 4.

\subsection{$\lambda_\to^*$ Abstraction} \label{exabst}

\begin{figure}
  \begin{CenteredBox}
\begin{lstlisting}[style=leanHH]
|!map!| : ∀ {α β : Type}, (α → β) → List α → List β
|!reverse!| : ∀ {α : Type}, List α → List α
|!map_reverse!| : ∀ {α β : Type} (f : α → β) (l : List α),
  map f (reverse l) = reverse (map f l)
|!reverse_reverse!| : ∀ {α : Type} (as : List α),
  reverse (reverse as) = as
⊢ ∀ (A B : Type) (f : A → B) (xs : List A),
    reverse (map f (reverse xs)) = map f xs
\end{lstlisting}
  \end{CenteredBox}
  \caption{Lean 4 proof state of a problem involving \textbf{List}.} \label{leanlistpretty}
\end{figure}

The Lean 4 proof state of the problem we will consider is shown in Figure \ref{leanlistpretty}.
The hypotheses (premises) and variable declarations are displayed before $\vdash$, while the goal comes after $\vdash$.
\vmaprev states that \vmap commutes with \vrev, and
\vrevrev states that \vrev is the inverse function of itself.

\begin{figure}
  \begin{CenteredBox}
\begin{lstlisting}[style=leanHH]
|!map!| : ∀ {α β : Type}, (α → β) → List α → List β
|!reverse!| : ∀ {α : Type}, List α → List α
|!map_reverse!| : ∀ {α β : Type} (f : α → β) (l : List α),
  @Eq (List β) (@map α β f (@reverse α l)) (@reverse β (@map α β f l))
|!reverse_reverse!| : ∀ {α : Type} (as : List α),
  @Eq (List α) (@reverse α (@reverse α as)) as
A B : Type
f : A → B
xs : List A
|!neg_goal!| : Not (@Eq (List B)
  (@reverse B (@map A B f (@reverse A xs))) (@map A B f xs))
⊢ False
\end{lstlisting}
  \end{CenteredBox}
  \caption{Lean 4 proof state after variable introduction and application of
    proof by contradiction, with implicit arguments displayed.
    Note that the equality sign in Figure \ref{leanlistpretty} is syntactic sugar for
    the polymorphic function \vEq shown here.}
  \label{leanlistexplicit}
\end{figure}

Since the problem is already in the $\lambda C$ fragment of Lean, the only
preprocessing step required is to introduce the universal quantifiers
appearing in the goal into the context and then apply proof by contradiction. The resulting proof state
is shown in Figure \ref{leanlistexplicit}.
For clarity, we have displayed the implicit arguments of all the functions.

First, we focus on translating \vneggoal into $\text{HOL}^*$. Following the
discussion in Sect. \ref{subencmon}, we would like to find a $\text{HOL}^*$ formula $\varphi$
and a ``substitution'' $\sigma$ such that the image of $\varphi$ under $\sigma$ is \vneggoal.
We also want the problem to be provable after the translation, so $\varphi$
should preserve as much information in \vneggoal as possible.

Three polymorphic functions: \vEq, \vmap and \vrev, occur in \vneggoal.
Although these functions are polymorphic, instances of these functions
with their dependent arguments instantiated behave like $\text{HOL}^*$ variables
(we will refer to such instances as $\mathit{HOL}^*$ \textit{instances}).
The type constructor \texttt{List} is also not allowed in $\text{HOL}^*$, but
\texttt{List A} and \texttt{List B} behave just like $\text{HOL}^*$ type variables
(we will refer to expressions such as \texttt{List A} and \texttt{List B} as $\mathit{HOL}^*$ \textit{type instances}).
Therefore, we can choose
$$\begin{aligned}
  \varphi := & \ \neg (\mathsf{EqLB} \ (\mathsf{rB} \ (\mathsf{mAB} \ f^* \ (\mathsf{rA} \ \mathit{xs}^*))) \ (\mathsf{mAB} \ f^* \ \mathit{xs}^*)) \\
  \sigma := & \ \{\mathsf{EqLB} \mapsto \texttt{@Eq (List B)}, \ \ \mathsf{mAB} \mapsto \texttt{@map A B}, \\
            & \ \ \mathsf{rA} \mapsto \texttt{@reverse A}, \ \ \mathsf{rB} \mapsto \texttt{@reverse B}, \ \ f^* \mapsto \texttt{f}, \ \ \mathit{xs}^* \mapsto \texttt{xs} \\
            & \ \ \mathsf{LA} \to \texttt{List A}, \ \ \mathsf{LB} \to \texttt{List B}, \ \ \mathsf{A} \to \texttt{A}, \ \ \mathsf{B} \to \texttt{B}\}, \\
\end{aligned}$$
where $\mathsf{EqLB} : \mathsf{LB} \to \mathsf{LB} \to \mathsf{Bool}, \
\mathsf{rA} : \mathsf{LA} \to \mathsf{LA}, \ \mathsf{rB} : \mathsf{LB} \to \mathsf{LB}, \ 
\mathsf{mAB} : (\mathsf{A} \to \mathsf{B}) \to \mathsf{LA} \to \mathsf{LB}, \
f^* : \mathsf{A} \to \mathsf{B}, \ \mathit{xs}^* : \mathsf{LA}$.

In a sense, the $\text{HOL}^*$ (type) instances are ``abstracted'' to $\text{HOL}^*$ (type)
variables. Note that the logical rules of $\text{HOL}^*$ are not relevant to this abstraction procedure---only
the term calculus $\lambda_\to^*$ is involved. Therefore, we name this procedure $\lambda_\to^*$ \textit{abstraction}.

However, $\lambda_\to^*$ abstraction is not directly applicable to \vmaprev and
\vrevrev, because dependent arguments of polymorphic functions occurring in them contain
universally quantified variables. Naturally, we would like to instantiate the quantifiers to make $\lambda_\to^*$
abstraction applicable.

\subsection{Quantifier Instantiation} \label{exinst}

To understand how quantifiers should be instantiated, we investigate how they would
be instantiated if we were to prove the goal manually. There are at least two ways we can proceed. We can
either first use \texttt{@map\_reverse A B} to swap the outer \texttt{reverse} with \texttt{map}, then
use \texttt{@reverse\_reverse A} to eliminate \texttt{reverse}; or, first use
\texttt{@map\_reverse A B} to swap the inner \texttt{reverse} with \texttt{map}, then
use \texttt{@reverse\_reverse B} to eliminate \texttt{reverse}. Notice how the dependent
arguments of a function $f$\footnote{In the context of this problem, $f$
could be \texttt{reverse} or \texttt{map}.} in the instantiated hypotheses
match the dependent arguments of $f$ in the $\text{HOL}^*$ instances of $f$ in the goal.

Quantifier instantiation in Lean-auto's monomorphization procedure is based
on a matching procedure that reflects the above observation. Given a set of formulas $S$,
the matching procedure first computes the set $M$ of $\text{HOL}^*$ instances occurring
in $S$ and then matches expressions in $S$ with elements of $M$. For example,
given $S=$\texttt{ \{@map\_reverse, @reverse\_reverse, neg\_goal\}}, the set $M$ is
\texttt{\{@reverse A, @reverse B, @map A B, @Eq (List B)\}}, all of whose
elements are collected from \vneggoal.
The matching procedure will preform the following matchings:

\begin{enumerate}
  \item \texttt{@Eq (List $\beta$)} in \vmaprev with \texttt{@Eq (List B)},
    which produces \texttt{fun $\alpha$ => @map\_reverse $\alpha$ B}
  \item \texttt{@map $\alpha$ $\beta$} in \vmaprev with \texttt{@map A B},
    which produces \texttt{@map\_reverse A B}
  \item \texttt{@reverse $\alpha$} in \vmaprev with \texttt{@reverse A} and \texttt{@reverse B},
    which produces \texttt{@map\_reverse A} and \texttt{@map\_reverse B}
  \item \texttt{@reverse $\beta$} in \vmaprev with \texttt{@reverse A} and \texttt{@reverse B},
    which produces \texttt{fun $\alpha$ => @map\_reverse $\alpha$ A} and \texttt{fun $\alpha$ => @map\_reverse $\alpha$ B}
  \item \texttt{@Eq (List $\alpha$)} in \vrevrev with \texttt{@Eq (List B)},
    which produces \texttt{@reverse\_reverse B}
  \item \texttt{@reverse $\alpha$} in \vrevrev with \texttt{@reverse A} and \texttt{@reverse B},
    which produces \texttt{@reverse\_reverse A} and \texttt{@reverse\_reverse B}
\end{enumerate}

Since \texttt{@reverse\_reverse A}, \texttt{@reverse\_reverse B} and \texttt{@map\_reverse A B}
are present, the instances produced are already sufficient for proving the goal.
But generally speaking, newly generated hypothesis instances and $\text{HOL}^*$
instances\footnote{New $\text{HOL}^*$ instances are collected from newly generated hypothesis instances.}
can still be matched with each other (and existing ones) to produce new useful results.
Hence, Lean-auto's monomorphization uses a saturation loop which
repeats the matching procedure until either no new instances can be produced or a prescribed
threshold is reached.

\subsection{Challenges Related to Dependent Type Theory and Lean 4} \label{exqdet}

\subsubsection{Dependent Arguments are Dynamic:}

\begin{figure}
  \begin{CenteredBox}
    \begin{lstlisting}[style=leanHH]
|!@DFunLike.coe!| : {F : Type (max u_1 u_5)}
  → {α : outParam (Type u_1)} → {β : outParam (α → Type u_5)}
  → [self : DFunLike F α β] → F → (a : α) → β a

@DFunLike.coe (A₀ →+ B₀) A₀ (fun x => B₀) AddMonoidHom.instFunLike f₀ a 
    \end{lstlisting}
  \end{CenteredBox}
  \caption{The function \texttt{DFunLike.coe} from MathLib4 and an expression
  containing it.}
  \label{dfun}
\end{figure}

  In $\lambda C$, whether an argument is dependent depends on how previous arguments
are instantiated. Consider the example shown in Figure \ref{dfun}. Here \texttt{DFunLike.coe}
is a low-level utility which turns a function-like object into its corresponding
function. In the signature of \texttt{DFunLike.coe}, the return type
\texttt{$\beta$ a} depends on the last argument \texttt{a : $\alpha$}. However, when
$\beta$ is instantiated with \texttt{fun x => B$_0$}, as in the
expression at the bottom of Figure \ref{dfun}, the return type \texttt{$\beta$ a} reduces to \texttt{B$_0$},
which no longer depends on the last argument. Our monomorphization procedure takes preceding arguments into
consideration when determining whether an argument is dependent.

\subsubsection{$\text{HOL}^*$ Instances are Dynamic:}
  In $\lambda C$, whether an expression is a $\text{HOL}^*$ instance is also context-dependent.
Consider the simple expression \texttt{@reverse = @reverse}, where \texttt{reverse}
is the same as in Figure \ref{leanlistexplicit}. Although \texttt{@reverse} is polymorphic,
it \textit{behaves like} a $\text{HOL}^*$ variable in \texttt{@reverse = @reverse}. More formally,
let
$$\begin{aligned}
\varphi &:= (f = f) \\
\sigma  &:= \{f \mapsto \texttt{@reverse}, \ \ \gamma \mapsto \texttt{(} \forall \ \texttt{\{}\alpha \ \beta \ : \ \texttt{Type\}}, \ 
  \texttt{List} \ \alpha \to \texttt{List} \ \beta \texttt{)}\}
\end{aligned}$$
where $f : \gamma$. Then, \texttt{@reverse = @reverse} is the image of the $\text{HOL}^*$ formula $\varphi$
under $\sigma$. Intuitively, the dependent arguments in the type of \texttt{reverse} can be ``absorbed''
into the $\text{HOL}^*$ type variable $\gamma$ because neither of the dependent arguments of \texttt{reverse} are present.
Our monomorphization procedure is able to detect such context-dependent $\text{HOL}^*$ instances.

\subsubsection{Definitional Equality:}
  As mentioned before, two syntactically different expressions can be definitionally
equal in Lean 4. Somehow, we need to account for this in Lean-auto's translation. Theoretically speaking, reducing all expressions to normal forms would solve
the problem to a large extent. However, full reduction is prohibitively
expensive on complex expressions in real-life Lean 4 projects, and the reduced expressions
could be much larger than the original expressions.\footnote{\refappendix{ssred} presents
a set of experiments that demonstrate these issues.}
Moreover, the reduced expressions might contain complex dependent types that Lean-auto cannot handle.
Therefore, we devise several other methods to address definitional equality.
  
  In Lean-auto, there are three separate occasions where definitional equality
has to be addressed.

  First, when a symbol is defined in Lean 4, (potentially multiple) \textit{equational theorems} that
reflect the definitional equalities related to the symbol are automatically generated.
Lean-auto can be configured to collect these equational theorems and to use
them to perform reduction and unfold constants (see Sect. \ref{sectprep}).

  Second, during $\lambda_\to^*$ abstraction, we would like $\text{HOL}^*$ instances
that are syntactically different but definitionally equal to be abstracted to the
same $\text{HOL}^*$ variable. Our $\lambda_\to^*$ abstraction algorithm keeps a
set $H$ of mutually definitionally unequal $\text{HOL}^*$ instances. Whenever a new $\text{HOL}^*$
instance $t$ is found, we test definitional equality of $t$ with elements of $H$
using \texttt{isDefEq}. Since \texttt{isDefEq} is expensive,
a \textit{fingerprint}\footnote{Roughly speaking, a \textit{fingerprint} of an expression
is a summary of the expression's syntax.} is computed for each $\text{HOL}^*$ instance, and fingerprint equality
is tested before calling \texttt{isDefEq}.

  Finally, even if two $\text{HOL}^*$ instances are definitionally unequal, there could still be nontrivial
relations between them. For example, if $f : \mathbb{N} \to \mathbb{N}$ is defined as
$f := \lambda (x : \mathbb{N}). g \ x \ x$, the equation $\forall (x : \mathbb{N}). f \ x = g \ x \ x$
would be a nontrivial relationship between $f$ and $g$. Lean-auto will attempt
to generate such equational theorems during quantifier instantiation. For each pair of $\text{HOL}^*$ instances $c_1, c_2$,
Lean-auto attempts to find terms $t_1, \dots, t_n$ such that $\lambda x_1 \dots x_m. \ c_1 \ y_1 \ \dots \ y_l = c_2 \ t_1 \ \dots \ t_n$,
where $x_1, \dots, x_m$ are variables occurring in $t_1, \dots, t_n$, and
$\{y_1, \dots, y_l\}$ is a subset of $\{x_1, \dots, x_m\}$.

\subsubsection{Absorbing Typeclass Instance Arguments:}
  In Lean 4, many functions have instance arguments that are not dependent arguments.
An example is the fourth argument of \texttt{HAdd.hAdd} mentioned in Sect.~\ref{sectlean}.
Since instance arguments are usually large expressions synthesized by Lean 4's typeclass
inference algorithm, translating them can result in large $\text{HOL}^*$ terms.
Lean-auto's implementation attempts to absorb typeclass arguments into
$\text{HOL}^*$ variables by instantiating typeclass instance quantifiers
and requiring $\text{HOL}^*$ instances to take typeclass arguments with them.\footnote{For
simplicity, this detail is not discussed in \refappendix{labstalgo} and~\ref{appinstant}.}

\section{$\lambda_\to^*$ Abstraction}\label{sectabst}

In this section, we discuss the $\lambda_\to^*$ abstraction procedure,
the second step of Lean-auto's monomorphization. Note that universe lifting, the first step,
is presented in \refappendix{appulift}. As mentioned before, we use $\Gamma \vdash? p$ to
represent the \emph{problem} of finding a proof of $p$ under context $\Gamma$.

The goal of $\lambda_\to^*$ abstraction is to translate essentially higher-order
problems (EHOPs) into $\text{HOL}^*$. Intuitively, a $\lambda C$ problem $\Gamma \vdash? p$ is EHOP iff there exists
a provable $\text{HOL}^*$ problem $\Gamma' \vdash? p'$ and a ``substitution'' $\sigma$ such that $\Gamma \vdash? p$ is the
image of $\Gamma' \vdash? p'$ under $\sigma$. Given $\Gamma \vdash? p$,
$\lambda_\to^*$ abstraction attempts to find such a triple $(\Gamma', p', \sigma)$. The formal definition of EHOP relies on
the concept of $\text{HOL}^*$-to-$\lambda C$ substitution and canonical embedding (see \refappendix{fmehop}).

\begin{definition}\label{EHOPInMain} A $\lambda C$ problem $\Gamma \vdash?  p$ is essentially higher-order provable (EHOP)
  iff there exists a provable $\text{HOL}^*$ problem $\Gamma' \vdash? p'$ and a substitution
  $(\pi^*(\Gamma'), \Gamma, \sigma)$ such that $p \cong \overline{\sigma}(\pi^*(p'))$.
\end{definition}

As a practical algorithm, Lean-auto's $\lambda_\to^*$ abstraction only works
on input problems $\Gamma \vdash? p$ where $p$ is a $\lambda C$ term structurally similar to $\text{HOL}^*$ terms. We
call such $\lambda C$ terms \textit{quasi-monomorphic terms}. They serve as the
intermediate representation between quantifier instantiation and $\lambda_\to^*$
abstraction. We use $\mathsf{QMono}(\Gamma; B, t)$ to represent ``$t$ is quasi-monomorphic
under context $\Gamma$, with variables in $B$ being bound variables.''\footnote{See \refappendix{labstalgo} for
the formal definition of $\mathsf{QMono}$.}
$\mathsf{QMono}$ has the following properties:

\begin{enumerate}
  \item Canonically embedded $\text{HOL}^*$ terms are $\mathsf{QMono}$.
  \item In $\mathsf{QMono}$ terms, proofs cannot be bound by $\lambda$ or dependent $\forall$ binders.
  \item A dependently typed free variable does not break the $\mathsf{QMono}$ property iff
    its dependent arguments do not contain bound variables.
  \item A dependently typed bound variable does not break the $\mathsf{QMono}$ property iff
    its dependent arguments are not instantiated.
  \item Except for within type declarations of bound variables, bodies of $\forall$ abstractions must be propositions.
\end{enumerate}

The $\lambda_\to^*$ abstraction algorithm itself is conceptually simple, but
it involves many technical details because it must handle all possible
features of $\mathsf{QMono}$ terms.\footnote{See \refappendix{labstalgo} for details of the algorithm.}
Given a $\lambda C$ problem $\Gamma \vdash? p$,
the $\lambda_\to^*$ abstraction algorithm traverses $p$ and turns $\text{HOL}^*$ instances it finds into $\text{HOL}^*$
variables. The ``substitution'' it returns is the map from $\text{HOL}^*$ variables to their
corresponding $\text{HOL}^*$ instances.

\section{Quantifier Instantiation}\label{sectinst}
In this section, we discuss the first step of Lean-auto's monomorphization : quantifier instantiation.
Given a context $\Gamma$ and a list of hypotheses $h_1 : t_1, \dots, h_n : t_n$, the quantifier
instantiation procedure of Lean-auto attempts to instantiate quantifiers in
$t_1, \dots, t_n$ to obtain terms suitable for $\lambda_\to^*$ abstraction
(i.e., to obtain terms that satisfy the $\mathsf{QMono}$ predicate).

As mentioned in Sect. \ref{exinst}, quantifier instantiation
is based on a saturation loop which matches $\text{HOL}^*$ instances of functions
with subterms of hypothesis instances. There are two main algorithms
in quantifier instantiation: $\textsf{matchInst}$ and $\textsf{saturate}$.
The $\textsf{matchInst}$ algorithm is responsible for matching $\text{HOL}^*$
instances with subterms of hypothesis instances to generate new
hypothesis instances, and the $\textsf{saturate}$
algorithm is the main saturation loop. The $\textsf{saturate}$ algorithm
is given in Algorithm \ref{saturateInMain}.\footnote{See \refappendix{appinstant}
for the $\textsf{matchInst}$ algorithm and details of the $\textsf{saturate}$ algorithm.}

\begin{algorithm}\label{saturateInMain}
  \DontPrintSemicolon
  \SetNoFillComment
  \SetKwFunction{matchOnePairFun}{\textsf{matchOnePair}}
  \SetKwFunction{saturateFun}{\textsf{saturate}}
  \caption{Main saturation loop of quantifier instantiation}
  \Fn{\saturateFun{$\Gamma; H, \vmaxInsts$}}{
    \Input{$\lambda C$ context $\Gamma$, list of $\lambda C$ terms $H$, and threshold $\vmaxInsts$}
    \Output{A list of $\lambda C$ terms}
    $\vhi := H$ \tcc*[h]{A list of hypothesis instances} \;
    $\vci := \mathsf{List.empty}()$ \tcc*[h]{A list of constant instances} \;
    \tcc*[h]{A queue of active constant and hypothesis instances} \;
    $\vactive := \mathsf{Queue.empty}()$\;
    \For{h : H}{
      $\vhi.\mathsf{push}((0, h))$ \;
      \For{$c : \mathsf{holInsts}(\Gamma; \emptyset, h)$}
        {$\vci.\mathsf{push}(c)$; $\vactive.\mathsf{push}((1, c))$}
    }
    \While{$! \ \vactive.\mathsf{empty}()$}{
      \lIf{$\vhi.\mathsf{size}() + \vci.\mathsf{size}() > \vmaxInsts$}{\Break}
      $(\vtype, \vfront) := \vactive.\mathsf{front}()$ \;
      $\vactive.\mathsf{popFront}()$ \;
      \eIf{$\vtype = 0$}{
        $\vprevci := \vci.\mathsf{copy}()$ \;
        \For{$c : \vprevci$}{
          $\mathsf{matchOnePair}(c, \vfront, \vci, \vhi, \vactive)$
        }
      }{
        $\vprevhi := \vhi.\mathsf{copy}()$ \;
        \For{$h : \vprevhi$}{
          $\mathsf{matchOnePair}(\vfront, h, \vci, \vhi, \vactive)$
        }
      }
    }
    $\vmonohi := \mathsf{List.empty}()$ \;
    \For{$h : \vhi$}{
      \lIf{$\mathsf{QMono}(\Gamma; \emptyset, h)$}{$\vmonohi.\mathsf{push}(h)$}
    }
    \Return $\vmonohi$
  }
  \;
  \Fn{\matchOnePairFun{$c, h, \vci, \vhi, \vactive$}}{
    $\vnewhi := \mathsf{matchInst}(\Gamma; c, h)$ \;
    \For{$\vnh : \vnewhi$}{
      \lIf{$\vnh \in \vhi$}{\Continue}
      $\vhi.\mathsf{push}(\vnh); \vactive.\mathsf{push}((0, \vnh))$ \;
      $\vnewci := \mathsf{holInsts}(\Gamma; \emptyset, \vnh)$ \;
      \For{$\vnc : \vnewci$}{
        \lIf{$\vnc \in \vci$}{\Continue}
        $\vci.\mathsf{push}(\vnc); \vactive.\mathsf{push}((1, \vnc))$ \;
      }
    }
  }
\end{algorithm}

The $\textsf{saturate}$ algorithm maintains a queue of active $\text{HOL}^*$ instances
and hypothesis instances, denoted as $\mathit{active}$. In each loop, an
element is popped from $\mathit{active}$. If it is a $\text{HOL}^*$ instance,
it is matched with all existing hypothesis instances; if it is a hypothesis instance,
it is matched with all existing $\text{HOL}^*$ instances. For each newly
generated hypothesis instance $h$, both $h$ and all the $\text{HOL}^*$
instances occurring in $h$ are added to $\mathit{active}$.

The $\textsf{saturate}$ algorithm also handles equational theorem generation of $\text{HOL}^*$
instances.\footnote{For simplicity, equational theorem generation is not shown
in Algorithm \ref{saturateInMain}.} For each new $\text{HOL}^*$ instance $c$, we generate
equational theorems between $c$ and existing $\text{HOL}^*$ instances.
The newly generated equational theorems are added to the set of existing hypothesis
instances so that they can participate in later matchings.

\section{Preprocessing}\label{sectprep}

  Preprocessing translates Lean 4 into dependent type theory, with
  the exception that part of definitional equality handling happens during monomorphization.
  In this section, we list the major steps of Lean-auto's preprocessing.


\subsubsection{Definitional Equality:}

  To handle definitional equality in Lean 4, Lean-auto partially reduces the input
  expressions, using Lean 4's built-in \texttt{Meta.transform} and \texttt{Meta.whnf}.
  This includes $\beta \zeta \eta \iota$ reduction and part of $\delta$
  reduction. In Lean 4, $\delta$ reduction is controlled by a reducibility setting,
  and Lean-auto allows users to specify the reducibility setting used by the preprocessor.
  For finer-grained control over which constants should be unfolded, Lean-auto allows
  users to supply a \textit{definitional equality instruction} $d[g_1, \dots, g_n]$
  and an \textit{unfolding instruction} $u[f_1, \dots, f_n]$, where $f_i, g_i$
  are constants.

  For the definitional equality instruction, Lean-auto automatically collects all the definitional
  equalities associated with $g_1, \dots, g_n$ and combines them with the premises
  supplied by the user. For the unfolding instruction, Lean-auto recursively unfolds $f_1, \dots, f_n$.
  To ensure termination, Lean-auto performs a topological sort on $f_1, \dots, f_n$,
  where $f_i$ is sorted before $f_j$ if $f_j$ occurs in the definition of $f_i$.
  Lean-auto will fail if there is a cyclic dependency between $f_1, \dots, f_n$.

  The preprocessing stage also performs equational theorem generation. It collects
  all maximal subexpressions of the input that do not contain logical symbols,
  and generates equational theorems between them. These equational theorems are
  also added to the list of premises.

\subsubsection{Inductive Types:}

  Currently, Lean-auto supports polymorphic, nested, and mutual inductive types
  when SMT solvers are used as the backend ATP. For other ATPs or unsupported
  inductive types, users can always manually supply the properties related
  to the inductive types as a workaround.

  The translation procedure for inductive types resembles monomorphization. For a polymorphic inductive type
  $T : \forall (\alpha_1 : \mathsf{U}_{\ell_1}) \ \dots \ (\alpha_n : \mathsf{U}_{\ell_n}). \mathsf{U}_\ell$,
  the translation attempts to find all relevant instances $T \ \alpha_1 \ \dots \ \alpha_n$,
  and translates each instance to a monomorphic inductive type in the SMT solver.
  For mutual and nested inductive types, the type of their constructors might contain
  other inductive types not occurring in the input premises. These inductive
  types will be recursively collected and monomorphized by the translation procedure.

\subsubsection{Quantifier Introduction and Proof by Contradiction:}

  To prepare for monomorphization, Lean-auto performs quantifier introduction on the goal and
  applies proof by contradiction. Suppose the goal is $\forall (x_1 : \alpha_1) \ \dots \ (x_n : \alpha_n). \beta$. 
  Quantifier introduction will introduce $x_1 : \alpha_1, \dots, x_n : \alpha_n$ into the context and replace
  the goal with $\beta$. Then, proof by contradiction will introduce the negation of the
  the goal $h : \neg \beta$ into the context and replace the goal with $\bot$.

\section{Experiments}\label{sectexpr}

  We evaluate Lean-auto and existing tools on user-declared theorems in
  Mathlib4,\footnote{Commit 29f9a66d622d9bab7f419120e22bb0d2598676ab.}
  using version \texttt{leanprover/lean4:v4.15.0} of Lean 4.
  A Lean 4 constant is considered a user-declared theorem if it is marked as a
  theorem, is declared somewhere in a \texttt{.lean} file,\footnote{We use
  \texttt{Lean.findDeclarationRanges?} to test whether a theorem is
  declared in a \texttt{.lean} file.} and is not a projection function.
  Due to technical reasons,\footnote{Refer to \refappendix{desetup}.}
  27762 of the 176904 user-declared theorems
  are excluded in our evaluation. Therefore, our benchmark set
  consists of 149142 theorems (problems). Evaluation is conducted on an Amazon
  EC2 \texttt{c5ad.16xlarge} instance with 64 CPU cores and 128GB memory. Each theorem
  is given a time limit of 10 seconds. Technical details of our
  experimental setup are discussed in \refappendix{desetup}.
  
  Since our primary goal is to evaluate Lean-auto's translation procedure,
  we do not use premise selection in our evaluation. Instead, for each theorem $T$ used
  in the evaluation, we collect all the theorems used in $T$'s human proof, and
  send them to Lean-auto and existing tools as premises.
  This simple procedure emulates an ideal premise selection algorithm.

  Three types of ATPs are used together with Lean-auto:
  \begin{enumerate}
    \item Native provers, or ATPs implemented in Lean 4 itself. Currently, the only general-purpose
      native prover supported by Lean-auto is Duper \cite{DuperPaper}. Although Duper can accept Lean 4 problems directly,
      it has difficulty handling Lean 4 features such as typeclasses and
      definitional equality. Our small-scale experiment shows that Duper only works
      well when used as a backend of Lean-auto.\footnote{Refer to \refappendix{ssduper}.}
      Considering that we also encountered technical issues when we attempted full-scale evaluation
      using Duper without Lean-auto, we decided to not include ``Duper without Lean-auto''
      in our evaluation.
    \item TPTP solvers. We chose Zipperposition, a higher-order superposition prover.
      Lean-auto sends problems to Zipperposition in TPTP TH0 format.
    \item SMT solvers. For this category, we chose Z3 and CVC5. Since SMT solvers
      still don't fully support HOL, we implemented a slightly modified version of the monomorphization
      procedure which generates FOL output. The modification introduces some extra incompleteness
      to the translation, which might have given Z3 and CVC5 a slight disadvantage.
  \end{enumerate}

  Currently, Lean-auto only supports proof reconstruction for native provers,
  utilizing a verified checker implemented in Lean-auto. The
  independent ongoing project Lean-smt\footnote{GitHub link:
  https://github.com/ufmg-smite/lean-smt} aims
  to support SMT proof reconstruction in the future.

  We compare Lean-auto with the following existing tools:
  \begin{enumerate}
    \item Lean 4's built-in tactic \texttt{rfl}. The \texttt{rfl} tactic
      proves theorems of the form \texttt{lhs = rhs} where \texttt{lhs}
      is definitionally equal to \texttt{rhs}. Note that \texttt{rfl}
      does not accept premises.
    \item Lean 4's built-in tactic \texttt{simp\_all}. Similar to Lean-auto,
      \texttt{simp\_all} accepts a list of user-provided premises. In Lean 4, users can
      tag theorems with the ``simp'' attribute. The \texttt{simp\_all} tactic
      succeeds on a decent portion of Mathlib4 even if we do not supply it with premises,
      because it has access to the theorems tagged with the ``simp'' attribute,
      and will use these theorems to simplify the input expressions.
      Therefore, we evaluate \texttt{simp\_all} in two different ways: with premises (``simp\_all''
      in Figure \ref{figevalcmp}) and without premises (``simp\_all - p'' in Figure \ref{figevalcmp}).
    \item The rule-based proof search procedure Aesop \cite{AesopPaper}. Since Aesop
      invokes the \texttt{simp\_all} tactic during its execution, it also
      benefits from theorems tagged with ``simp.'' We
      evaluate Aesop in two different ways: with premises\footnote{Specifically, for
      each premise $p$, we add \texttt{(add unsafe p)} to the \texttt{aesop} invocation} (``Aesop''
      in Figure \ref{figevalcmp}) and without premises (``Aesop - p'' in Figure \ref{figevalcmp}).
  \end{enumerate}

  Due to limited time and resources, this work does not compare Lean-auto with
  hammers implemented in other ITPs. Differences in logical systems make it
  very difficult to translate datasets between ITPs. For example, even though
  Lean 4 and Coq are both based on dependent type theory, they extend
  dependent type theory in different ways.\footnote{For example, Cumulative Universe Levels in Coq
  and Quotient Types in Lean 4.} Translation procedures between Lean 4 and
  Coq would need to modify expressions in nontrivial ways, which would cause
  typechecking and definitional equality issues.

  \begin{figure}
  \begin{center}\begin{tabular}{| l | K{6em} K{7em} K{7em} |}
    \hline
                      & Solved & Unique Solves  & Avg Time(ms) \\ \hline
    rfl               & 19896  & 35             & 5.7          \\
    simp\_all - p     & 9833   &                & 19.8         \\
    simp\_all         & 28096  &                & 52.0         \\
    simp\_all VBS     & 28204  & 3035           & 44.4         \\
    Aesop - p         & 33762  &                & 61.3         \\
    Aesop             & 47060  &                & 93.5         \\
    Aesop VBS         & 48413  & 6512           & 92.2         \\
    Lean-auto + Duper & 54570  &                & 1092.5       \\
    Lean-auto + Z3    & 54210  &                & 863.5        \\
    Lean-auto + CVC5  & 54316  &                & 808.0        \\
    Lean-auto + Zipper. &54817 &                & 774.9        \\
    Lean-auto VBS     & 61906  & 22020          & 756.8        \\ \hline
    Overall VBS       & 79396  &                & 314.7        \\ \hline
  \end{tabular}\end{center}
  \caption{Comparison with existing tools. Our benchmark set contains 149135 problems.} \label{figevalcmp}
  \end{figure}

  \begin{figure}
  \begin{center}
    \includegraphics[width=120mm]{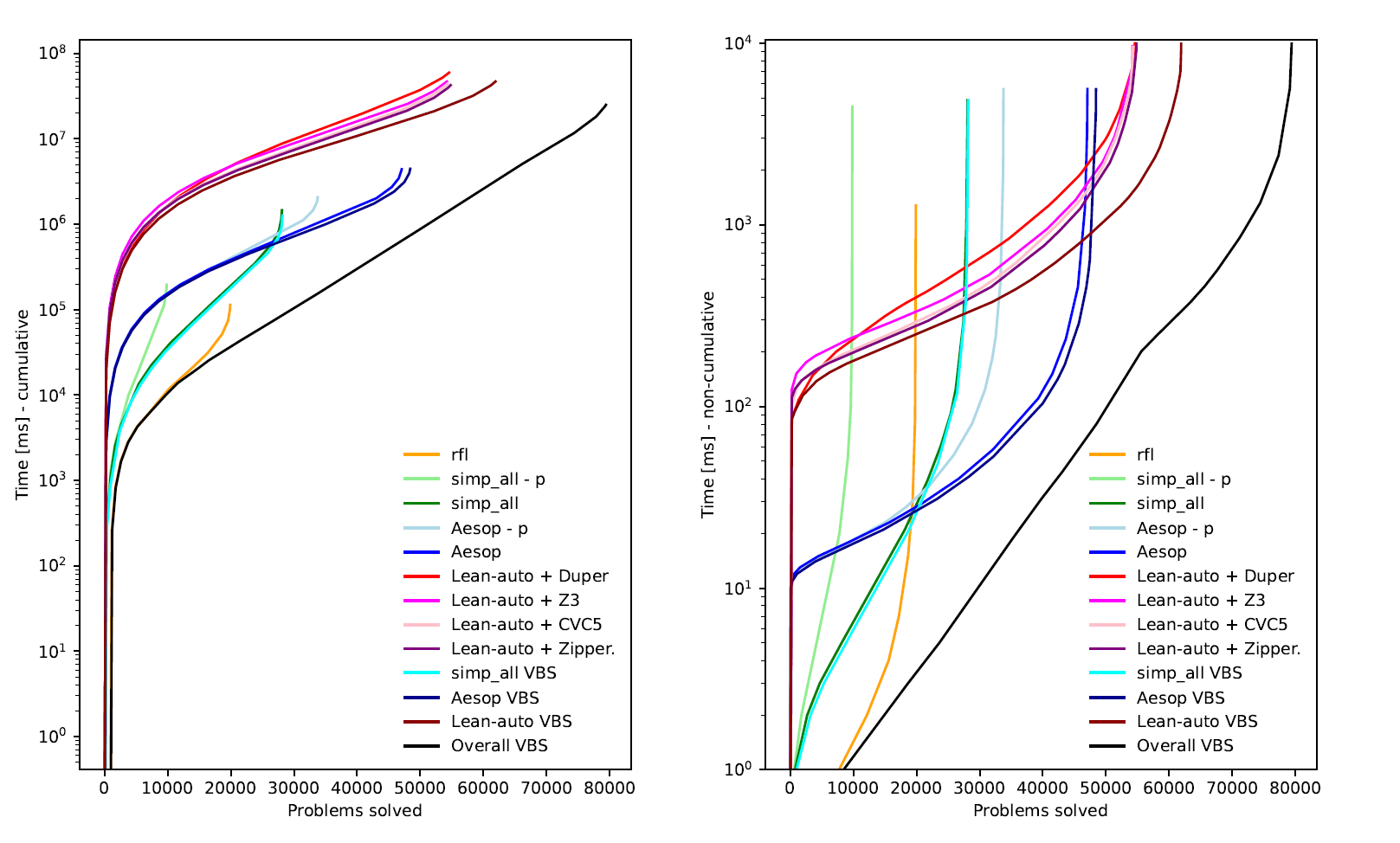}
  \end{center}
  \caption{\#Solved - Cumulative Time plot (left) and \#Solved - Time plot (right)}
  \label{figsolvedtime}
  \end{figure}

  Results are shown in Figure \ref{figevalcmp}. For ``simp\_all'', ``aesop,''
  and ``Lean-auto,'' we show the results of their virtual best solvers (VBSes).\footnote{The
  virtual best solver of a given category is equivalent to running
  all the tools in the given category in parallel and taking the first success produced.} We compute
  unique solves among ``rfl'' and these three VBSes.
  
  We find that Lean-auto solves more problems than all existing tools. Specifically,
  ``Lean-auto + Duper'', which supports proof reconstruction, solves 36.6\%
  problems in our benchmark set, which is 5.0\% better than the best previous tool
  ``Aesop''. The fact that ``Lean-auto VBS'' achieves 14.8\% unique solves
  shows that Lean-auto is complementary to existing tools. The overall VBS,
  which combines Lean-auto and all existing tools, solves more than half (53.2\%) of
  the problems in our benchmark set. On the other hand, Lean-auto is
  significantly slower than existing tools on solved problems. This is potentially caused by
  Lean-auto's verified checker and the frequent definitional equality
  testing in Lean-auto's monomorphization.

  To better compare the performance of the various tools, we plot, for each tool, the
  number of solved problems vs. solving time and cumulative solving time.
  The results are shown in Figure \ref{figsolvedtime}. We see that
  Lean-auto is slower than existing tools on simple problems, but eventually
  solves more problems than all existing tools.

\section{Conclusion}

In this paper, we presented the ITP to ATP translation implemented in Lean-auto.
Our contributions are three-fold. First, we addressed challenges posed by
Lean 4's dependent type theory and its various language features. Second, we 
designed a novel monomorphization procedure for dependent type theory.
Finally, we implemented the translation procedure in Lean-auto and evaluated it on Mathlib4.

A possible direction for future work is to design a complete $\lambda_\to^*$
abstraction algorithm. Another direction is to investigate
potential ways of handling existential type quantifiers and non-leading universal
type quantifiers. We would also like to further investigate causes of
Lean-auto's inefficiencies and improve Lean-auto's performance.
\begin{credits}
  \subsubsection{\ackname} The authors thank: Prof. Jasmin Blanchette (Ludwig Maximilian University of Munich)
  for insightful discussions on the monomorphization procedure in Isabelle Sledgehammer; Mario
  Carneiro (Chalmers University of Technology) for helping us understanding implementation
  details of Lean 4; and Leonardo de Moura (Amazon Web Services) for his advice on
  the translation from Lean 4 to SMT solvers. We also greatly appreciate the help of the
  Lean Zulip users who answered our questions related to Lean 4 and Mathlib4.
  This work was supported in part by the Stanford Graduate Fellowship, the Stanford Center for Automated Reasoning,
  and AFRL and DARPA under Agreement FA8750-24-9-1000.

  \subsubsection{\discintname} Clark Barrett is an Amazon Scholar.
  
\end{credits}

%
%
%
%
  %
  %
  %
  %

\bibliographystyle{splncs04}
\bibliography{references}

\appendix

\maybeappendix{Logical Symbols of $\lambda C$}\label{applsc}
\ifdefined\cameraReady
\else
  \begin{align*}
  \bot &:= \forall p : \mathsf{U}_0. p \ \ \ (\neg) := \lambda p : \mathsf{U}_0. p \to \bot \\
  (\land) &:= \lambda p \ q : \mathsf{U}_0. \forall r : \mathsf{U}_0. (p \to q \to r) \to r \\
  (\lor) &:= \lambda p \ q : \mathsf{U}_0. \forall r : \mathsf{U}_0. (p \to r) \to (q \to r) \to r \\
  (\leftrightarrow) &:= \lambda p \ q. (p \to q) \land (q \to p) \\
  (=_\ell) &:= \lambda \alpha : \mathsf{U}_\ell. \lambda x \ y : \alpha. \forall p : \alpha \to \mathsf{U}_0. (p \ x \leftrightarrow p \ y) \\
  (\exists_\ell) &:= \lambda \alpha : \mathsf{U}_\ell. \lambda p : \alpha \to \mathsf{U}_0. \forall q : \mathsf{U}_0. ((\forall x : \alpha. p \ x \to q) \to q)
  \end{align*}
\fi

\maybeappendix{Derivation Rules of PTS}\label{apppts}
\ifdefined\cameraReady
\else
  The type judgement $\Gamma \vdash t : \alpha$ in a PTS specified
  by $(\mathcal{S}, \mathcal{A}, \mathcal{R})$ is defined by the following
  axioms and rules:

  \begin{align*}
  & \text{(axioms)}      & & \frac{}{\emptyset \vdash s_1 : s_2} & & \text{if } (s_1, s_2) \in \mathcal{A} \\
  & \text{(start)}       & & \frac{\Gamma \vdash A : s}{\Gamma, x : A \vdash x : A} & & \text{if } x \notin \Gamma \\
  & \text{(weakening)}   & & \frac{\Gamma \vdash A : B \ \ \ \Gamma \vdash C : s}{\Gamma, x : C \vdash A : B} & & \text{if } x \notin \Gamma \\
  & \text{(product)}     & & \frac{\Gamma \vdash A : s_1 \ \ \ \Gamma, x : A \vdash B : s_2}{\Gamma \vdash (\forall x : A. B) : s_3} & & \text{if } (s_1, s_2, s_3) \in \mathcal{R} \\
  & \text{(application)} & & \frac{\Gamma \vdash f : (\forall x : A. B) \ \ \ \Gamma \vdash a : A}{\Gamma \vdash f \ a : B[x := a]} & & \\
  & \text{(abstraction)} & & \frac{\Gamma, x : A \vdash b : B \ \ \ \Gamma \vdash (\forall x : A. B) : s}{\Gamma \vdash (\lambda x : A.b) : (\forall x : A. B)} & & \\
  & \text{(conversion)}  & & \frac{\Gamma \vdash A : B \ \ \ \Gamma \vdash B' : s \ \ \ B \cong B'}{\Gamma \vdash A : B'} & &
  \end{align*}
\fi

\maybeappendix{$\lambda C, \lambda_\to$ and $\lambda_\to^*$}\label{applll}
\ifdefined\cameraReady
\else
  \begin{definition} $\lambda C$ is the pure type system $(\mathcal{S}, \mathcal{A}, \mathcal{R})$ where
    $$\mathcal{S} := \{\mathsf{U}_\ell | \ell \in \mathbb{N}\} \ \ \ \mathcal{A} := \{(\mathsf{U}_\ell, \mathsf{U}_{\ell + 1}) | \ell \in \mathbb{N}\}$$
    $$\mathcal{R} := \{(\mathsf{U}_\ell, \mathsf{U}_m, \mathsf{U}_{\mathsf{imax}(\ell, m)}) | \ell \in \mathbb{N}, m \in \mathbb{N}\}$$
    $$\mathsf{imax}(m, n) := \left\{\begin{aligned}
      \mathsf{max}(m, n), & & n > 0 \\
      0, & & n = 0
    \end{aligned}\right.$$
  \end{definition}
  
  \begin{definition} $\lambda_\to$ is the pure type system $(\mathcal{S}, \mathcal{A}, \mathcal{R})$ where
    $$\mathcal{S} := \{\mathsf{U}_1, \mathsf{U}_1'\} \ \ \ \mathcal{A} := \{(\mathsf{U}_1, \mathsf{U}_1')\} \ \ \ 
      \mathcal{R} := \{(\mathsf{U}_1, \mathsf{U}_1, \mathsf{U}_1)\}$$
    This is equivalent to simply typed lambda calculus, where $\mathsf{U}_1$ and $\mathsf{U}_1'$ are
    usually denoted as $*$ and $\square$, respectively.
  \end{definition}

  \begin{definition} $\lambda_\to^*$ is the pure type system $(\mathcal{S}, \mathcal{A}, \mathcal{R})$ where
    $$\mathcal{S} := \{\mathsf{U}_\ell | \ell \in \mathbb{N}^*\} \cup \{\mathsf{U}_\ell' | \ell \in \mathbb{N}^*\} \ \ \
      \mathcal{A} := \{(\mathsf{U}_\ell, \mathsf{U}_\ell') | \ell \in \mathbb{N}^*\}$$
    $$\mathcal{R} := \{(\mathsf{U}_\ell, \mathsf{U}_m, \mathsf{U}_{\mathsf{max} \{l, m\}}) | \ell \in \mathbb{N}^*, m \in \mathbb{N}^*\}$$
  \end{definition}
\fi

\maybeappendix{HOL and $\text{HOL}^*$}\label{apphol}
\ifdefined\cameraReady
\else
  \begin{definition} $\text{HOL}$ ($\text{HOL}^*$) is defined as $\lambda_\to$ ($\lambda_\to^*$) augmented with the following symbols:
    \begin{enumerate}
      \item $\mathsf{Bool}$
      \item $\bot'$ and $\to'$
      \item $\forall'_s$, for each $s \in \mathcal{T}_\to^*$. Note that we are not requiring $s$ to be a type here
        because the typing rules below will ensure that $s$ must be a type in a well-formed $\forall_s'$.
    \end{enumerate}
    
    \noindent the following typing rules:
    $$\frac{}{\vdash \mathsf{Bool} : \mathsf{U}_1} \ \ \ \ \frac{}{\Gamma \vdash \bot' : \mathsf{Bool}}$$
    $$\frac{}{\Gamma \vdash \to' : \mathsf{Bool} \to \mathsf{Bool} \to \mathsf{Bool}} \ \ \ \
    \frac{\Gamma \vdash s : \mathsf{U}_\ell}{\Gamma \vdash \forall'_s : (s \to \mathsf{Bool}) \to \mathsf{Bool}}$$
    
    \noindent and the logical axioms and deduction rules of higher-order logic.

  \end{definition}
  
  \noindent \textbf{Note:} The logical symbols $\neg', \land', \lor', \leftrightarrow, ='_s, \exists'_s$ are defined in a way consistent
  with their definition in $\lambda C$:
  \begin{align*}
  (\neg') &:= \lambda (p : \mathsf{Bool}). (p \to' \bot') \\
  (\land') &:= \lambda (p \ q : \mathsf{Bool}). \forall (r : \mathsf{Bool}). ((p \to' q \to' r) \to' r) \\
  (\lor') &:= \lambda (p \ q : \mathsf{Bool}). \forall (r : \mathsf{Bool}). ((p \to' r) \to' (q \to' r) \to' r) \\
  (\leftrightarrow') &:= \lambda (p \ q : \mathsf{Bool}). ((p \to' q) \land' (q \to' p)) \\
  (='_s) &:= \lambda (x \ y : s). \forall (p : s \to \mathsf{Bool}). (p \ x \leftrightarrow' p \ y) \\
  (\exists'_s) &:= \lambda (p : s \to \mathsf{Bool}). \forall (q : \mathsf{Bool}). ((\forall (x : s). p \ x \to' q) \to' q)
  \end{align*}

  We use $\forall' (x : s). t$ as a shorthand for $\forall'_s \ (\lambda (x : s). t)$, and $\exists'(x : \alpha). t$
  as a shorthand for $\exists_s' \ (\lambda (x : s). t)$.

  For simplicity, the $\text{HOL}^*$ system we present in this paper
  only contains one symbol $\mathsf{Bool}$ for the type of propositions. In the implementation
  of Lean-auto, the $\text{HOL}^*$ system have a symbol $\mathsf{Bool}_\ell : \mathsf{U}_\ell$
  for each universe level $\ell$, and each universe level have its own copy of logical symbols.
\fi

\maybeappendix{Universe Lifting}\label{appulift}
\ifdefined\cameraReady
\else
  In this appendix, we discuss the translation procedure from $\text{HOL}^*$ to HOL
  in Lean-auto. For simplicity, universe lifting as presented in this
  section differs from Lean-auto's implementation in terms of
  how $\mathsf{Bool}$ is handled.
  
  First, we show that $\text{HOL}^*$ and HOL are, in a sense, equivalent to each other.

  \begin{definition}
    Let $\rho^* : \mathcal{T}_\to^* \to \mathcal{T}_\to$ be the mapping that forgets the universe levels, i.e.
    $$\begin{aligned}
    & \rho^*(\mathsf{Bool}) := \mathsf{Bool} \ \ \ \ \ \ \rho^*(\mathsf{U}_\ell) := \mathsf{U}_1 \ \ \ \ \ \
      \rho^*(\mathsf{U}_\ell') := \mathsf{U}_1' \ \ \ \ \ \ \rho^*(x) := x, \text{ for }x \in V \\
    & \rho^*(M \ N) := \rho^*(M) \ \rho^*(N) \ \ \ \ \ \ \rho^*(\lambda (x : s). M) := \lambda (x : \rho^*(s)). \rho^*(M) \\
    & \rho^*(\bot') := \bot' \ \ \ \ \ \ \rho^*(\to') := \to' \ \ \ \ \ \ \rho^*(\forall_s') := \forall_{\rho^*(s)}'
    \end{aligned}$$

    \noindent $\rho^*$ is extended to contexts as follows: $\rho^*(\emptyset) := \emptyset; \ \rho^*(\Gamma, x : \sigma) := \rho(\Gamma), x : \rho(\sigma)$
  \end{definition}

  \begin{definition}
    Let $\rho_\ell : \mathcal{T}_\to \to \mathcal{T}_\to^*(\ell \in \mathbb{N}^*)$ be the mapping that turns $\mathsf{U}_1$ into $\mathsf{U}_\ell$, i.e.
    $$\begin{aligned}
    & \rho(\mathsf{Bool}) := \mathsf{Bool} \ \ \ \ \ \ \rho(\mathsf{U}_1) := \mathsf{U}_\ell \ \ \ \ \ \
      \rho(\mathsf{U}_1') := \mathsf{U}_\ell' \ \ \ \ \ \ \rho(x) := x, \text{ for }x \in V \\
    & \rho(M \ N) := \rho(M) \ \rho(N) \ \ \ \ \ \ \rho(\lambda (x : s). M) := \lambda (x : \rho(s)). \rho(M) \\
    & \rho(\bot') := \bot' \ \ \ \ \ \ \rho(\to') := \to' \ \ \ \ \ \ \rho(\forall_s') := \forall_{\rho(s)}'
    \end{aligned}$$

    \noindent $\rho$ is extended to contexts as follows: $\rho(\emptyset) := \emptyset; \ \rho(\Gamma, x : \sigma) := \rho(\Gamma), x : \rho(\sigma)$
  \end{definition}

  \begin{theorem}
    For all $t \in \mathcal{T}_\to$, $\rho_\ell^*(\rho_\ell(t)) = t$.
    \begin{proof} Induction on the construction rules of $\mathcal{T}_\to$. \end{proof}
  \end{theorem}

  \begin{theorem}
    Forgetting universe levels preserves judgement, i.e., if $\Gamma \vdash t : s$ in $\text{HOL}^*$,
    then $\rho^*(\Gamma) \vdash \rho^*(t) : \rho^*(s)$ in $\text{HOL}$.
    \begin{proof} Induction on the derivation rules of $\text{HOL}^*$. \end{proof}
  \end{theorem}

  \begin{theorem}
    $\rho_\ell$ preserves judgement, i.e., if $\Gamma \vdash t : s$ in HOL, then
    $\rho_\ell(\Gamma) \vdash \rho_\ell(t) : \rho_\ell(s)$ in $\text{HOL}^*$.
    \begin{proof} Induction on the derivation rules of HOL. \end{proof}
  \end{theorem}

  \begin{theorem}
    $\text{HOL}^*$ and HOL are equivalent, i.e., if $\Gamma \vdash p : \mathsf{Bool}$ in $\text{HOL}^*$
    and $p$ is provable in $\text{HOL}^*$, then $\rho^*(p)$ is provable in HOL; if $\Gamma \vdash p : \mathsf{Bool}$
    in HOL and $p$ is provable in HOL, then $\rho_\ell(p)$ is provable in $\text{HOL}^*$ for any $\ell \in \mathbb{N}^*$.
    \begin{proof}
      Let $\mathcal{D}$ be a proof of $p$ in $\text{HOL}^*$, the a proof of
      $\rho^*(p)$ in HOL can be obtained by forgetting universe levels in $\mathcal{D}$.
      The converse can be proved in a similar way.
    \end{proof}
  \end{theorem}

  \noindent The universe lifting procedure in Lean-auto is the translation of $\text{HOL}^*$
  to HOL in the context of $\lambda C$. In other words, it is the translation of \textit{the}
  embedding of HOL in $\lambda C$ into \textit{an} embedding of $\text{HOL}^*$ in $\lambda C$.

  \begin{definition}
    The $\ell$-embedding $\pi_\ell : \mathcal{T}_\to \to \mathcal{T}_C$ of HOL into
    $\lambda C$ is defined as $\pi_\ell := \pi^* \circ \rho_\ell$, where $\pi^*$ is the
    canonical embedding of $\text{HOL}^*$ into $\lambda C$.\footnote{See \refappendix{fmehop}
    for the definition of $\pi^*$.}
  \end{definition}

  \begin{definition}
    A universe lifting facility consists of three families of functions
    \begin{enumerate}
      \item $\mathsf{GLift}_{u, v} : \mathsf{U}_u \to \mathsf{U}_{\max\{u, v + 1\}}$
      \item $\mathsf{GLift.up}_{u, v} : \forall (\alpha : \mathsf{U}_u). \ \alpha \to \mathsf{GLift}_{u, v} \ \alpha$
      \item $\mathsf{GLift.down}_{u, v} : \forall (\alpha : \mathsf{U}_u). \ \mathsf{GLift}_{u, v} \ \alpha \to \alpha$
    \end{enumerate}
    where $u, v \in \mathbb{N}$, such that they satisfy the following bijectivity condition:
    $$\forall (\alpha : \mathsf{U}_u). \mathsf{GLift.up}_{u, v} \ \alpha \circ \mathsf{GLift.down}_{u, v} \ \alpha = \lambda (x : \mathsf{GLift}_{u, v} \ \alpha). x$$
    $$\forall (\alpha : \mathsf{U}_u). \mathsf{GLift.down}_{u, v} \ \alpha \circ \mathsf{GLift.up}_{u, v} \ \alpha = \lambda (x : \alpha). x$$
  \end{definition}

  \noindent In Lean 4, universe lifting facility can be realized by the following inductive type:
  
  \begin{lstlisting}[style=leanHH]
structure GLift.{u, v} (α : Sort u) : Sort (max u (v + 1)) where
  /-- Lift a value into `GLift α` -/    up ::
  /-- Extract a value from `GLift α` -/ down : α
  \end{lstlisting}

  \begin{theorem} Assume the existence of a universe lifting facility in $\lambda C$. Then,
    for all $\ell \in \mathbb{N}$, there exists two families of $\lambda C$ functions
    $$\mathsf{Up}_s : s \to \mathsf{UpType} \ s \ \ \ \mathsf{Down}_s : \mathsf{UpType} \ s \to s$$
    for sorts $s : \mathsf{U}_{\ell'}, \ell' \leq \ell + 1$, satisfying the bijectivity conditions
    $$\mathsf{Up}_s \circ \mathsf{Down}_s = \lambda (x : \mathsf{UpType} \ s). x \ \ \
      \mathsf{Down}_s \circ \mathsf{Up}_s = \lambda (x : s). s$$
    and the congruence condition
    $$\forall (f : \alpha \to \beta). \ \mathsf{Up}_\beta \ (f \ x) = (\mathsf{Up}_{\alpha \to \beta} \ f) \ (\mathsf{Up}_\alpha \ x)$$
    where $\mathsf{UpType}$ is recursively defined as follows:
    $$\begin{aligned}
    \mathsf{UpType} \ x & := \mathsf{GLift}_{\ell', \ell} \ x, \text{ for } x \in V \\
    \mathsf{UpType} \ (\alpha \to \beta) & := \mathsf{UpType} \ \alpha \to \mathsf{UpType} \ \beta
    \end{aligned}$$
    \begin{proof}
      Structural induction on $s$
      \begin{enumerate}
        \item If $s = a$, where $a \in V$ is a variable, then we can define
          $$\mathsf{Up}_a := \mathsf{GLift.up}_{\ell', \ell} \ a \ \ \ \mathsf{Down}_a := \mathsf{Glift.down}_{\ell', \ell} \ a$$
        \item If $s = (\alpha \to \beta)$ and the induction hypothesis holds for $\alpha$ and $\beta$, then we can define
          $$\begin{aligned}
          \mathsf{Up}_{\alpha \to \beta} &:= \lambda (f : \alpha \to \beta) \ (x : \mathsf{UpType} \ \alpha). \ \mathsf{Up}_\beta \ (f \ (\mathsf{Down}_\alpha \ x)) \\
          \mathsf{Down}_{\alpha \to \beta} &:= \lambda (f : \mathsf{UpType} \ \alpha \to \mathsf{UpType} \ \beta) \ (x : \alpha). \ \mathsf{Down}_\beta \ (f \ (\mathsf{Up}_\alpha \ x))
          \end{aligned}$$
      \end{enumerate}
    \end{proof}
  \label{uliftupdownthm}
  \end{theorem}

  The rationale of $\mathsf{UpType} \ (\alpha \to \beta) := \mathsf{UpType} \ \alpha \to \mathsf{UpType} \ \beta$
  is that, given $f \ x$ in the canonical embedding of $\text{HOL}^*$, where $f : \alpha \to \beta$ and
  $x : \alpha$, we would like $(\mathsf{Up}_{\alpha \to \beta} \ f) \ (\mathsf{Up}_\alpha \ x)$ to be type correct.

  Given $\mathsf{Up}_s$, $\mathsf{Down}_s$ and $\mathsf{UpType}$ satisfying Theorem \ref{uliftupdownthm}
  with $\ell$ taken to be larger than all universe levels in the input terms,
  the universe lifting procedure, or the translation of the canonical embedding of $\text{HOL}^*$ into
  the $\ell$-embedding of HOL, denoted as $\mathsf{ULiftTrans}$, works as follows:
  
  \begin{enumerate}
    \item If $x$ is a variable and $x : s$, then $\mathsf{ULiftTrans}(x) := x'$.
      If $x$ is not a free variable, then we define $x'$ as $x' := \mathsf{Up}_s \ x$ in Lean 4.
      If $x$ is a bound variable, no further operation is needed.
    \item $\mathsf{ULiftTrans}(f \ x) := \mathsf{ULiftTrans}(f) \ \mathsf{ULiftTrans}(x)$
    \item $\mathsf{ULiftTrans}(\lambda (x : s).\ y) := \lambda (x' : \mathsf{UpType} \ s).\ \mathsf{ULiftTrans}(y)$
  \end{enumerate}

  It's easy to verify that $\mathsf{ULiftTrans}(e)$ is definitionally equal to $\mathsf{Up}_s \ e$
  for all terms $e : s$ in the canonical embedding of $\text{HOL}^*$.
\fi

\maybeappendix{Essentially Higher-order Problem}\label{fmehop}
\ifdefined\cameraReady
\else
  In this appendix, we give a formal definition of essentially higher-order problems (EHOPs) and
  discuss some of its theoretical properties.
  
  \begin{definition} Let $\sigma : V \to \mathcal{T}_C$ be a mapping.
    Define its extension $\overline{\sigma} : \mathcal{T}_C \to \mathcal{T}_C$ as
    $$\overline{\sigma}(\mathsf{U}_\ell) := \mathsf{U}_\ell \ \ \ \ \ \
      \overline{\sigma}(x) := \sigma(x), \text{ for }x \in V \ \ \ \ \ \
      \overline{\sigma}(M \ N) := \overline{\sigma}(M) \ \overline{\sigma}(M)$$
    $$\overline{\sigma}(\lambda x : s. M) := \lambda x : \overline{\sigma}(s). \overline{\sigma[x \mapsto x]}(M)$$
    $$\overline{\sigma}(\forall x : s. M) := \forall x : \overline{\sigma}(s). \overline{\sigma[x \mapsto x]}(M)$$
    where
    $$\sigma[u \to t](x) := \left\{\begin{aligned}
      & t & & x = u \\
      & \sigma(x) & & x \in V \backslash \{u\}
    \end{aligned}\right.$$
  \end{definition}
  
  \begin{definition}
    A substitution is a triple $(\Gamma, \Gamma', \sigma)$ where $\Gamma, \Gamma'$ are $\lambda C$ contexts
    and $\sigma : V \to \mathcal{T}_C$, such that for all $(u : \tau) \in \Gamma$,
    $$\Gamma' \vdash \sigma(u) : \overline{\sigma}(\tau)$$
    $\Gamma$ is called the domain of the substitution, and $\Gamma'$ is called the codomain of the substitution.
  \end{definition}
  
  \begin{theorem}
    Let $(\Gamma, \Gamma', \sigma)$ be a substitution. If $\Gamma \vdash t : s$, then $\Gamma' \vdash \overline{\sigma}(t) : \overline{\sigma}(s)$
  \end{theorem}
  \begin{proof} Induction on the derivation of $\Gamma \vdash t : s$.
  \end{proof}
  
  \begin{definition} Let $\Gamma$ be a $\lambda C$ context and $t_1, t_2$ be $\lambda C$ terms.
    If variable set $M$ and substitution $(\Gamma, \Gamma', \sigma)$ satisfies
    \begin{enumerate}
      \item There exists a $\lambda C$ term $s$ such that $\Gamma' \vdash \overline{\sigma}(t_1) : s$ and $\Gamma' \vdash \overline{\sigma}(t_2) : s$.
      \item $\overline{\sigma}(t_1) \cong \overline{\sigma}(t_2)$ (i.e., $\overline{\sigma}(t_1)$ and
        $\overline{\sigma}(t_2)$ are $\beta\eta$-equivalent)
      \item For all variables $v \in \Gamma \backslash M$, $\sigma(v) = v$.
    \end{enumerate}
    Then $(\Gamma, \Gamma', \sigma)$ is called a $M$-unifier of $t_1$ and $t_2$. In the context of Lean,
    this corresponds to a unifier of $t_1$ and $t_2$ under context $\Gamma$, with $M$ as the set of metavariables.
  \end{definition}
  
  \begin{definition} The canonical embedding $\pi^* : \mathcal{T}_\to^* \to \mathcal{T}_C$ of $\text{HOL}^*$ into $\lambda C$ is defined as follows:
    $$\begin{aligned}
    & \pi^*(\mathsf{Bool}) := \mathsf{U}_0 \ \ \ \ \ \ \pi^*(\mathsf{U}_\ell) := \mathsf{U}_\ell \ \ \ \ \ \
      \pi^*(\mathsf{U}_\ell') := \mathsf{U}_{\ell + 1} \ \ \ \ \ \ \pi^*(x) := x, \text{ for } x \in V \\
      & \pi^*(M \ N) := \pi^*(M) \ \pi^*(N) \ \ \ \ \ \ \pi^*(\lambda (x : s). M) := \lambda (x : \pi^*(s)). \pi^*(M) \\
    & \pi^*(\bot') := \forall (\alpha : \mathsf{U}_0). \alpha \ \ \ \ \ \ \ \ \ \
    \pi^*(\to') := \lambda (p \ q : \mathsf{U}_0). p \to q \\
    & \pi^*(\forall_s') := \lambda (p : \pi^*(s) \to \mathsf{U}_0). \forall (x : \pi^*(s)). p \ x
    \end{aligned}$$
  
    \noindent $\pi^*$ is extended to contexts as follows: $\pi^*(\emptyset) := \emptyset, \pi^*(\Gamma, x : \sigma) := \pi^*(\Gamma), x : \pi^*(\sigma)$
  
  
  \end{definition}
  
  \begin{theorem}\label{ceptj} Canonical embedding preserves judgement, i.e. if $\Gamma \vdash t : s$ in $\text{HOL}^*$, then
    $\pi^*(\Gamma) \vdash \pi^*(t) : \pi^*(s)$ in $\lambda C$ \end{theorem}
  \begin{proof} Induction on the derivation rules of $\text{HOL}^*$. \end{proof}
  
  \begin{definition} An ($\text{HOL}^*/\lambda C$) problem is a tuple $(\Gamma, p)$, denoted
    as $\Gamma \vdash? p$, where $\Gamma$ is a ($\text{HOL}^*/\lambda C$)
    context, called the hypotheses of the problem, and $p$ is an
    ($\text{HOL}^*/\lambda C$) term, called the goal of the problem. A $\lambda C$ problem
    $\Gamma \vdash? p$ is provable iff there exists a $\lambda C$ term $t$ such that
    $\Gamma \vdash t : p$. An $\text{HOL}^*$ problem $\Gamma \vdash? p$ is provable iff there exists
    a $\lambda C$ term $t$ such that $\pi^*(\Gamma) \vdash t : \pi^*(p)$.
  \end{definition}
  
  \begin{definition} A $\lambda C$ problem $\Gamma \vdash?  p$ is essentially higher-order provable (EH\-OP)
    iff there exists a provable $\text{HOL}^*$ problem $\Gamma' \vdash? p'$ and a substitution
    $(\pi^*(\Gamma'),\allowbreak\Gamma, \sigma)$ such that $p \cong \overline{\sigma}(\pi^*(p'))$.\footnote{This
    is the same as Definition \ref{EHOPInMain} in Sect. \ref{sectabst}.}
  \end{definition}
  
  \begin{theorem}
    If a $\lambda C$ problem $\Gamma \vdash? p$ is EHOP, then it is provable.
  \end{theorem}
  \begin{proof} By the definition of EHOP, there exists a provable $\text{HOL}^*$ problem
    $\Gamma' \vdash? p'$ and substitution $(\pi^*(\Gamma'), \Gamma, \sigma)$ such that
    $p \cong \overline{\sigma}(\pi^*(p'))$. By the definition of $\text{HOL}^*$ provability, there exists
    a term $t'$ such that $\pi^*(\Gamma') \vdash t' : \pi^*(p')$. By Theorem \ref{ceptj},
    $\Gamma \vdash \overline{\sigma}(t') : \overline{\sigma}(\pi^*(p'))$, thus $\Gamma \vdash? p$
    is provable.
  \end{proof}
  
  We assume that excluded middle is implicitly contained in the hypotheses of all
  $\text{HOL}^*$ and $\lambda C$ problems. In $\lambda C$, excluded middle is
  $\mathsf{em} : \forall (p : \mathsf{U}_0), p \lor \neg p$; in $\text{HOL}^*$, it is
  $\mathsf{em}' : \forall (p : \mathsf{Bool}). p \lor' \neg' p$.
  
  \begin{example} Consider the $\lambda C$ problem $\Gamma \vdash? p$ where
  \begin{align*}
    \Gamma := \ & \mathbb{N} : \mathsf{U}_1, \mathsf{Fin} : \mathbb{N} \to \mathsf{U}_1,
    \mathsf{add} : \forall (n : \mathbb{N}). (\mathsf{Fin} \ n \to \mathsf{Fin} \ n \to \mathsf{Fin} \ n), n : \mathbb{N} \\
    p := \ & (\forall (u \ v : \mathsf{Fin} \ n). \mathsf{add} \ n \ u \ v =_1 \mathsf{add} \ n \ v \ u) \to \\
    & \ \ \ \forall (u \ v \ w : \mathsf{Fin} \ n). \mathsf{add} \ n \ (\mathsf{add} \ n \ x \ y) \ z =_1 \mathsf{add} \ n \ z \ (\mathsf{add} \ n \ y \ x)
  \end{align*}
  Given
  \begin{align*}
    \Gamma' := \ & \alpha : \mathsf{U}_1, f : \alpha \to \alpha \to \alpha \\
    p' := \ & (\forall' (u \ v : \alpha). f \ u \ v =_\alpha' f \ v \ u) \to' \\
    & \ \ \ \forall' (u \ v \ w : \alpha). f \ (f \ u \ v) \ w =_\alpha' f \ w \ (f \ v \ u)
  \end{align*}
  The $\text{HOL}^*$ problem $\Gamma' \vdash? p'$ is provable. Moreover, given
  $$\sigma(\alpha) := \mathsf{Fin} \ n, \sigma(f) := \mathsf{add} \ n$$
  The triple $(\pi^*(\Gamma'), \Gamma, \sigma)$ forms a substitution, and $p \cong \overline{\sigma}(\pi^*(p'))$.
  Therefore, $\Gamma \vdash? p$ is EHOP.
  \end{example}
  
  Note that moving implications in the goal into hypotheses (and vice versa) may
  change the EHOP status of a problem. For example,
  $$\alpha : \mathsf{U}_1, x : \alpha, p : \alpha \to \mathsf{U}_0  \ \vdash? \ (p \ x \to p \ x)$$
  is EHOP. However, if we introduce $p \ x$ into the hypotheses, the problem is no longer EHOP:
  \begin{equation}\label{hypnehop}
    \alpha : \mathsf{U}_1, x : \alpha, p : \alpha \to \mathsf{U}_0, h : p \ x  \ \vdash? \ p \ x
  \end{equation}
  
  \begin{theorem}
    The $\lambda C$ problem \eqref{hypnehop} is provable but not EHOP.
  \end{theorem}
  \begin{proof}
    Note that $h : p \ x$ under the hypotheses of \eqref{hypnehop}, thus \eqref{hypnehop} is provable.
    To show that \eqref{hypnehop} is not EHOP, we use proof by contradiction. Suppose there
    is an $\text{HOL}^*$ problem $\Gamma' \vdash? p'$ and a substitution $(\Gamma', \Gamma, \sigma)$ such
    that $p \ x \cong \overline{\sigma}(\pi^*(p'))$. Then, the $\beta\eta$ normal form of $p'$ must be of the form
    $f \ t_1 \ \dots \ t_k$ where $f$ is a free variable. Note that $\Gamma'$, as a context of $\lambda_\to^*$,
    consists solely of $\text{HOL}^*$ (type or term) variable declarations, and cannot contain
    premises like $\lambda C$ contexts. Note that there exists models where $f \ t_1 \ \dots \ t_k$
    is false, for example when $f$ is a function that takes $k$ arguments and always returns $\bot$.
    Therefore, $\Gamma' \vdash? p'$ is not provable in $\text{HOL}^*$, thus $(2)$ is not EHOP.
  \end{proof}
\fi

\maybeappendix{$\lambda_\to^*$ Abstraction Algorithm}\label{labstalgo}
\ifdefined\cameraReady
\else
  In this appendix, we give a formal presentation of the $\lambda_\to^*$ abstraction algorithm. When given a
  $\lambda C$ problem $\Gamma \vdash? p$, the algorithm attempts to find a $\lambda_\to^*$
  problem $\Gamma' \vdash? p'$ and a substitution $(\pi^*(\Gamma'), \Gamma, \sigma)$ such that
  $p \cong \overline{\sigma}(\pi^*(p'))$, and that $p'$ retains as much information in $p$ as possible.
  
  Note that the output of Lean-auto's quantifier instantiation is a list of $\lambda C$ terms
  $h_1, \dots, h_n$, and we would like to prove $\bot$ using these terms. Suppose
  the $\lambda C$ context of the problem is $\Gamma$. According to the above
  discussion, the input to the $\lambda_\to^*$ abstraction algorithm should be
  $\Gamma \vdash ? \ (h_1 \to \dots \to h_n \to \bot)$. In practice, we run
  $\lambda_\to^*$ abstraction consecutively on each of $h_i (1 \leq i \leq n)$ under
  context $\Gamma$, which produces equivalent results. Therefore, we can either think
  of the input of $\lambda_\to^*$ abstraction as one $\lambda C$ term $h_1 \to \dots \to h_n \to \bot$,
  or as a list of $\lambda C$ terms $h_1, \dots, h_n$.

  First, we give a formal definition of \textit{dependent arguments}. This definition accounts for
  the fact that dependent arguments are dynamic. Note that in the argument list of functions,
  dependent and non-dependent arguments may interleave with each other.
  
  \begin{definition} Suppose $\Gamma \vdash s : \mathsf{U}_l$ in $\lambda C$.
    If $s = (\forall (x : s_1). s_2)$ and $x$ occurs in $s_2$,
    then $s$ is said to be a $\Gamma$-leading argument dependent type,
    denoted as $\mathsf{LADT}(\Gamma; s)$. Suppose $\Gamma \vdash t : s$ in $\lambda C$, where $s$
    is in $\beta$ normal form. If $\mathsf{LADT}(\Gamma; s)$, then $t$ is said to be
    $\Gamma$-leading argument dependent ($\Gamma$-lad), denoted as $\mathsf{LAD}(\Gamma; t)$.
  \end{definition}
  
  \begin{definition} Suppose the term $a_0 \ a_1 \ \dots \ a_k$ is type correct under
    context $\Gamma$ in $\lambda C$. Then for $1 \leq i \leq k$, $a_0$ is said to have dependent $i$-th argument
    with respect to $\Gamma$ and argument list $(a_1, \dots, a_k)$, or $i$-dep w.r.t
    $\Gamma$ and $(a_1, \dots, a_k)$, iff $\mathsf{LAD}(\Gamma; a_0 \ a_1 \ \dots \ a_{i - 1})$.
    For convenience, we use the predicate
    $$\mathsf{Dep}(\Gamma; a_0, (a_1, \dots, a_k), i) \ \ \ (k \geq 0, 1 \leq i \leq k)$$
    to denote that $a_0$ is $i$-dep w.r.t $\Gamma$ and $(a_1, \dots, a_k)$. Furthermore, we define
    $$\mathsf{LFun}(\Gamma; a_0, (a_1, \dots, a_k)) := \lambda (x_{i_1} : s_{i_1}) \dots (x_{i_m} : s_{i_m}). a_0 \ w_1 \ \dots \ w_m$$
    $$\mathsf{DArgs}(\Gamma; a_0, (a_1, \dots, a_k)) := (b_{i_1}, \dots, b_{i_m})$$
    $$\mathsf{LArgs}(\Gamma; a_0, (a_1, \dots, a_k)) := (a_{j_1}, \dots, a_{j_{k - m}})$$
    where $i_1 < i_2 < \dots < i_m$ are all the arguments that are dependent, $j_1 < j_2 < \dots < j_{k - m}$
    are all the arguments that are non-dependent, $\Gamma \vdash a_i : s_i$, and
    $$w_i := \left\{\begin{aligned}
      a_i, & & \mathsf{Dep}(\Gamma; a_0, (a_1, \dots, a_k), i) \\
      x_i, & & \text{otherwise}
    \end{aligned}\right.$$
  \end{definition}
  
  \begin{example} Let
    \begin{align*}
      \Gamma := & \ \mathsf{compose} : \forall (\beta \ \gamma: \mathsf{U}_1).
        (\beta \to \gamma) \to \forall (\alpha : \mathsf{U}_1). (\alpha \to \beta) \to (\alpha \to \gamma), \\
        & \ A : \mathsf{U}_1, B : \mathsf{U}_1, C : \mathsf{U}_1, f : B \to C, g : A \to B, x : A
    \end{align*}
    Then
    $$\mathsf{compose}, \mathsf{compose} \ B, \mathsf{compose} \ B \ C \ f$$
    are $\Gamma$-lad, while
    $$\mathsf{compose} \ B \ C, \mathsf{compose} \ B \ C \ f \ A, \mathsf{compose} \ B \ C \ f \ A \ g$$
    are not. Therefore, the dependent arguments of $\mathsf{compose}$ w.r.t $(B, C, f, A, g, x)$
    are $1, 2$ and $4$, and we have
    $$\mathsf{LFun}(\Gamma; \mathsf{compose}, (A, B, C, f, g, x)) = \lambda (f : B \to C). \mathsf{compose} \ A \ B \ f \ C$$
    $$\mathsf{LArgs}(\Gamma; \mathsf{compose}, (A, B, C, f, g, x)) = (f, g, x)$$
  \end{example}
  
  \begin{example} Let
    \begin{align*}
      \Gamma := \mathsf{func} : \forall (\alpha : \mathsf{U}_1 \to \mathsf{U}_1) \ (\beta : \mathsf{U}_1). \alpha \ \beta,
        A : \mathsf{U}_1, B : \mathsf{U}_1 
    \end{align*}
    Then $\mathsf{func}$ is $\Gamma$-lad, while
    $$\mathsf{func} \ (\lambda \beta. A) : \mathsf{U}_1 \to A \ \ \ \ \ \
    \mathsf{func} \ (\lambda \beta. A) \ B : A$$
    are not. Therefore, the dependent argument of $\mathsf{func}$ w.r.t $(\lambda \beta. A, B)$ is $1$, and
    we have
    $$\mathsf{LFun}(\Gamma; \mathsf{func}, (\lambda \beta. A, B)) = \mathsf{func} \ (\lambda \beta . A) \ \ \ \ \ \
    \mathsf{LArgs}(\Gamma; \mathsf{func}, (\lambda \beta. A, B)) = B$$
  \end{example}
  
  Now, we define \textit{quasi-monomorphic terms}, the set of $\lambda C$ terms
  that $\lambda_\to^*$ abstraction can successfully translate to $\text{HOL}^*$. The
  predicate $\mathsf{QMono}(\Gamma; B, t)$ will be used to represent ``$t$ is a quasi-monomorphic
  term under context $\Gamma$, with variables in $B$ being bound variables''. It
  is used both in $\lambda_\to^*$ abstraction and in quantifier instantiation.

  \begin{definition} We define the predicate $\mathsf{QMono}(\Gamma; B, t)$ inductively,
    where $\Gamma$ is a $\lambda C$ context, $B$ is a set of variables, and $t$ is a $\lambda C$ term
    \begin{enumerate}
      \item For variable $x \in B$ and terms $t_1, \dots, t_n$,
        \begin{align*}
          \mathsf{QMono}(\Gamma; B, x \ t_1 \dots \ t_n) := \ &
          \mathsf{DArgs}(\Gamma; x, (t_1 \ \dots \ t_n)) = \emptyset \land \\
          & \forall i \in \{1, \dots, n\}. \mathsf{QMono}(\Gamma; B, t_i)
        \end{align*}
      \item For variable $x \notin B$ and terms $t_1, \dots, t_n$,
        \begin{align*}
          \mathsf{QMono}(\Gamma; B, x \ t_1 \dots \ t_n) := \ & (\forall t \in \mathsf{DArgs}(\Gamma; x, (t_1, \dots, t_n)). FV(t) \cap B = \emptyset) \land \\
          & (\forall t \in \mathsf{LArgs}(\Gamma; x, (t_1, \dots, t_n)). \mathsf{QMono}(\Gamma; B, t))
        \end{align*}
      \item For variable $x$ and terms $s, t$
        \begin{align*}
          \mathsf{QMono}(\Gamma; B, \lambda (x : s). t) := \
          & FV(s) \cap B = \emptyset \land (\Gamma \not\vdash s : \mathsf{U}_0) \\
          & \land \mathsf{QMono}(\Gamma, x : s; B \cup \{x\}, t)
        \end{align*}
      \item For variable $x$ and terms $s, t$ such that $x \in FV(t)$,
        \begin{align*}
          \mathsf{QMono}(\Gamma; B, \forall (x : s). t) := \
          & \neg FV(s) \cap B = \emptyset \land (\Gamma \not \vdash s : \mathsf{U}_0) \land  (\Gamma \vdash t : \mathsf{U}_0) \land \\
          & \mathsf{QMono}(\Gamma, x : s; B \cup \{x\}, t)
        \end{align*}
      \item For terms $s, t$,
        \begin{align*}
          \mathsf{QMono}(\Gamma; B, s \to t) := \ & (\Gamma \vdash s : \mathsf{U}_0) \land (\Gamma \vdash t : \mathsf{U}_0) \land \\
          & \mathsf{QMono}(\Gamma; B, s) \land \mathsf{QMono}(\Gamma; B, t)
        \end{align*}
    \end{enumerate}
  \end{definition}

  \noindent According to the definition of $\mathsf{QMono}$, terms coming from canonical embedding of $\text{HOL}^*$
  terms are automatically quasi-monomorphic, e.g.
  $$\mathsf{QMono}(\alpha : \mathsf{U}_1, p : (\alpha \to \alpha) \to \mathsf{U}_0; \emptyset, \forall (p : \alpha \to \alpha). f \ p)$$
  Proofs are not allowed to be quantified by $\lambda$ or dependent $\forall$ binders:
  $$\neg \mathsf{QMono}(p : \mathsf{U}_0, q : p \to \mathsf{U}_0; \emptyset, \forall (x : p). q \ x)$$
  Occurrence of a dependently typed free variable does not break the quasi-monomorphic property iff
  its dependent arguments do not contain bound variables (assuming $B = \emptyset$):
  \begin{align*}
    \mathsf{QMono}(
    & \mathbb{N} : \mathsf{U}_1, \mathsf{Fin} : \mathbb{N} \to \mathsf{U}_1,
      \mathsf{add} : \forall (n : \mathbb{N}). \mathsf{Fin} \ n \to \mathsf{Fin} \ n \to \mathsf{Fin} \ n, k : \mathbb{N}; \\
    & \emptyset, \forall (x \ y : \mathsf{Fin} \ k). \mathsf{add} \ k \ x \ y = \mathsf{add} \ k \ y \ x)
  \end{align*}
  Occurrence of a dependently typed bound variable does not break the quasi-monomorphic property iff
  its dependent arguments are not instantiated:
  $$\mathsf{QMono}(\emptyset; \emptyset, \lambda (f : (\forall (\alpha : \mathsf{U}_0). \alpha) \to (\forall (\alpha : \mathsf{U}_0). \alpha)) \
    (x : \forall (\alpha : \mathsf{U}_0). \alpha). f \ x)$$
  Except for within type declarations of bound variables, bodies of $\forall$ abstractions must be propositions:
  $$\neg \mathsf{QMono}(\alpha : \mathsf{U}_1, \beta : \alpha \to \mathsf{U}_1; \emptyset, \forall (x : \alpha). \beta \ x)$$

  \begin{algorithm}\label{lamabst}
    \DontPrintSemicolon
    \SetNoFillComment
    \SetKwFunction{lamAbstraction}{\textsf{lamAbst}}
    \SetKwFunction{gvn}{\textsf{getLVarName}}
    \caption{$\lambda_\to^*$ abstraction algorithm of Lean-auto}
    \Fn{\lamAbstraction{$\Gamma; B, t$}}{
      \Input{$\lambda C$ context $\Gamma$, variable set $B$,
        and $\lambda C$ term $t$ satisfying $\mathsf{QMono}(\Gamma; B, t)$}
      \Output{a $\lambda_\to^*$ term}
      \Switch(\textbf{with}){t}{
        \Case(\tcc*[h]{Function application}){$a \ b$}{
          $f := \mathsf{getAppFn}(t)$ \;
          $\vargs := \mathsf{getAppArgs}(t)$ \;
          \uIf{$f \in B$}{
            \For{$a : \vargs$}{$a := \mathsf{lamAbst}(\Gamma; B, a)$}
            \Return $\mathsf{mkAppN}(f, \vargs)$
          }
          $\vlf := \mathsf{LFun}(\Gamma; f, \vargs)$ \;
          $\vlargs := \mathsf{LArgs}(\Gamma; f, \vargs)$ \;
          $\vlvar := \mathsf{getLVarName}(\vlf)$ \;
          \Return $\mathsf{mkAppN}(\vlvar, \vlargs)$
        }
        \Case{$\forall (v : a). b$}{
          $\vatype := \mathsf{inferType}(\Gamma; a)$ \;
          $\vbabst := \mathsf{lamAbst}(\Gamma, v : a; B \cup \{v\}, b)$ \;
          \uIf{$\vatype = \mathsf{U}_0$}{
            $\vaabst := \mathsf{lamAbst}(\Gamma; B, a)$ \;
            \Return $\vaabst \to \vbabst$
          }
          \Return $\forall (v : a). \vbabst$
        }
        \Case{$\lambda (v : a). b$}{
          $\vbabst := \mathsf{lamAbst}(\Gamma, v : a; B \cup \{v\}, b)$ \;
          \Return $\lambda (v : a). \vbabst$
        }
        \Other{\Return $\mathsf{getLVarName}(t)$}
      }
    }
    \;
    \Fn{\gvn{t}} {
      \Input{$\lambda C$ term $t$}
      \Output{$\lambda_\to^*$ variable name corresponding to $t$}
      \uIf{$H.\mathsf{contains}(t)$}{
        \Return $H.\mathsf{find}(t)$
      }
      $\vnewname := \mathsf{freshVarName}()$ \;
      $H.\mathsf{add}(t, \vnewname)$ \;
      \Return $\vnewname$ 
    }
  \end{algorithm}
  
  Now, we describe the $\lambda_\to^*$ abstraction procedure $\mathsf{lamAbst}$ of Lean-auto. The algorithm
  is shown in Algorithm \ref{lamabst}. A global hash map $H$ is used to record the $\text{HOL}^*$
  variables associated with abstracted $\lambda C$ terms. A few auxiliary functions are used in the algorithm:
  \begin{enumerate}
    \item For a term $t$, if $t$ is in $H$, then $\mathsf{getLVarName}(t)$ returns
      the $\text{HOL}^*$ free variable corresponding to $t$, otherwise it creates a new $\text{HOL}^*$ free variable for $t$.
    \item For a term $t = w \ t_1 \ \dots \ t_n$ where $w$ is not an application,
      $\mathsf{getAppFn}(t) = w, \mathsf{getAppArgs}(t) = (t_1, \dots, t_n)$.
    \item For terms $w, t_1, \dots, t_n$, $\mathsf{mkAppN}(w, (t_1, \dots, t_n)) = w \ t_1 \ \dots \ t_n$.
    \item For a context $\Gamma$ and a term $t$, $\mathsf{inferType}(\Gamma, t)$ computes the
      $\beta$-normal form of the type of $t$ under $\Gamma$.
  \end{enumerate}
  Note that $\mathsf{lamAbst}$ only returns the $\text{HOL}^*$ problem (as a $\text{HOL}^*$ term). The
  ``substitution'' from $\text{HOL}^*$ to $\lambda C$ needs to be obtained by computing the inverse of $H$ after
  the execution of the algorithm. Also, note that the implementation of this algorithm in Lean-auto checks
  whether $t$ breaks the requirements of quasi-monomorphic-ness and fails if it
  does. For simplicity, these checks have been omitted in $\mathsf{lamAbst}$.
  
\fi

\maybeappendix{Quantifier Instantiation}\label{appinstant}
\ifdefined\cameraReady
\else
  In this appendix, we present the technical details of Lean-auto's quantifier
  instantiation procedure. First, we give a formal definition of \textit{instance}:
  
  \begin{definition}
    Let $\Gamma$ be a $\lambda C$ context, and $t$ be a $\lambda C$ term which is type correct under $\Gamma$.
    \begin{enumerate}
      \item A constant instance of $t$ is a $\lambda C$ term of the
        form $\lambda (x_1 : s_1) \dots (x_m : s_m). t \ t_1 \ \dots \ t_k$ that is
        type correct under $\Gamma$, where $s_1, \dots, s_m, t_1, \dots t_k$ are $\lambda C$ terms.
      \item For $t = \forall (x_1 : r_1) \dots (x_n : r_n). b$, a hypothesis instance of
        $t$ is a $\lambda C$ term of the form $\forall (y_1 : s_1) \dots (y_m : s_m). b[t_1/x_1]\dots[t_n/x_n]$,
        where $s_1, \dots, s_m, t_1, \dots, t_n$ are $\lambda C$ terms, and $t_1[t_2/x]$ stands
        for the term obtained by replacing all the $x$ in $t_1$ with $t_2$.
    \end{enumerate}
  \end{definition}
  
  Unless otherwise stated, when discussing instances of functions, we will always be
  referring to \textit{constant instances}; when discussing instances of hypotheses, we will
  always be referring to \textit{hypothesis instances}.
  An instance of a function is called an $\text{HOL}^*$ instance iff all of the function's dependent arguments are instantiated with terms
  that do not contain bound variables. Formally, the set of all $\text{HOL}^*$ instances
  in a $\lambda C$ term is defined as follows:
  
  \begin{definition}
    Let $\Gamma$ be a $\lambda C$ context and $B$ be a set of variables, then
    \begin{enumerate}
      \item For variable $x$ and terms $t_1, \dots, t_n$,
        $$\mathsf{holInsts}(\Gamma; B, x \ t_1 \ \dots \ t_n) := \left\{
          \begin{aligned}
            S \cup \{l\}, & & FV(l) \cap B = \emptyset \\
            S, & & \text{otherwise}
          \end{aligned}
        \right.$$
        where
        $$l := \mathsf{LFun}(\Gamma; x, (t_1 \ \dots \ t_n)) \ \ \ \ \ \ S := \bigcup_{t \in \mathsf{LArgs}(\Gamma; x, (t_1, \dots, t_n))} \mathsf{holInsts}(\Gamma; V, t)$$
      \item For variable $x$ and terms $a, b$,
        \begin{align*}
          \mathsf{holInsts}(\Gamma; B, \forall (x : a). b) = \mathsf{holInsts}(\Gamma; B, \lambda (x : a). b) 
          \\ := \mathsf{holInsts}(\Gamma; B, a) \cup \mathsf{holInsts}(\Gamma, x : a; B \cup \{x\}, b)
        \end{align*}
      \item Otherwise, $\mathsf{holInsts}(\Gamma; B, t) := \emptyset$.
    \end{enumerate}
  \end{definition}
  
  \begin{algorithm}\label{matching}
    \DontPrintSemicolon
    \SetNoFillComment
    \SetKwFunction{matchFun}{\textsf{match}}
    \caption{Matching algorithm for quantifier instantiation}
    \Fn{\matchFun{$\Gamma; M, m, h$}}{
      \Input{$\lambda C$ context $\Gamma$, variable set $M$, and $\lambda C$ terms $m, h$}
      \Output{A set of unifiers}
      \Switch(\textbf{with}){h}{
        \Case(\tcc*[h]{Function application}){$a \ b$}{
          $\vmatches := \emptyset$ \;
          $f := \mathsf{getAppFn}(t)$ \;
          $\vargs := \mathsf{getAppArgs}(t)$ \;
          \For{$a : \vargs$}{$\vmatches := \mathsf{union}(\vmatches, \mathsf{match}(\Gamma; M, m, a))$}
          $\vlf := \mathsf{LFun}(\Gamma; f, arg)$ \;
          $\vmatches := \mathsf{union}(\vmatches, \mathsf{unify}(\Gamma; M, m, \vlf))$
        }
        \Case{$\forall (v : a). b$}{
          \Return $\mathsf{union}(\mathsf{match}(\Gamma; M, m, a), \mathsf{match}(\Gamma, v : a; M, m, b))$
        }
        \Case{$\lambda (v : a). b$}{
          \Return $\mathsf{union}(\mathsf{match}(\Gamma; M, m, a), \mathsf{match}(\Gamma, v : a; M, m, b))$
        }
        \Other{\Return $\emptyset$}
      }
    }
  \end{algorithm}
  
  The matching procedure in the saturation loop is handled by $\mathsf{matchInst}$ and $\mathsf{match}$.
  \begin{enumerate}
    \item Given context $\Gamma$, variable set $M$ and terms $m, h$,
      $\mathsf{match}(\Gamma; M, m, h)$ returns all $M$-unifiers between term $m$ and the $\mathsf{LFun}$ of subterms of $h$.
      The pseudocode for $\mathsf{match}$ is given in Algorithm \ref{matching}. An auxiliary function
      $\mathsf{unify}$ is used in the pseudocode. Given $\lambda C$ context $\Gamma$, variable set $M$
      and two $\lambda C$ terms $t_1, t_2$, $\mathsf{unify}(\Gamma; M, t_1, t_2)$ returns a complete set of
      $M$-unifiers of $t_1$ and $t_2$ under $\Gamma$. In Lean 4, the \texttt{isDefEq} function can be used perform
      unification, but it is incomplete and returns at most one unifier.
    \item Given context $\Gamma$ and terms $m, h$,
      $\mathsf{matchInst}(\Gamma; m, h)$ computes all instances of the hypothesis $h$ which has some subterm whose
      $\mathsf{LFun}$ is $\beta\eta$-equivalent to $m$. To do this, $\mathsf{matchInst}$ introduces all leading non-prop $\forall$
      quantifiers into the context (as free variables), collects all the newly introduced free variables into a variable set $M$,
      then computes $\mathsf{match}(\Gamma'; M, m, h')$, where $\Gamma', h'$ are $\Gamma, h$ after introduction of free variables.
      For each unifier $(\Gamma, \Gamma', \sigma)$ in $\mathsf{match}(\Gamma'; M, m, h)$, $\mathsf{matchInst}$ computes $\overline{\sigma}(h)$,
      then abstracts newly introduced free variables in $\sigma$ as $\forall$ binders to generate an instance of $h$.
      $\mathsf{matchInst}(\Gamma; m, h)$ returns the set of instances of $h$ generated by this procedure. 
  \end{enumerate}
  
  \begin{algorithm}\label{saturate}
    \DontPrintSemicolon
    \SetNoFillComment
    \SetKwFunction{matchOnePairFun}{\textsf{matchOnePair}}
    \SetKwFunction{saturateFun}{\textsf{saturate}}
    \caption{Main saturation loop of quantifier instantiation}
    \Fn{\saturateFun{$\Gamma; H, \vmaxInsts$}}{
      \Input{$\lambda C$ context $\Gamma$, list of $\lambda C$ terms $H$, and threshold $\vmaxInsts$}
      \Output{A list of $\lambda C$ terms}
      $\vhi := H$ \tcc*[h]{A list of hypothesis instances} \;
      $\vci := \mathsf{List.empty}()$ \tcc*[h]{A list of constant instances} \;
      \tcc*[h]{A queue of active constant and hypothesis instances} \;
      $\vactive := \mathsf{Queue.empty}()$\;
      \For{h : H}{
        $\vhi.\mathsf{push}((0, h))$ \;
        \For{$c : \mathsf{holInsts}(\Gamma; \emptyset, h)$}
          {$\vci.\mathsf{push}(c)$; $\vactive.\mathsf{push}((1, c))$}
      }
      \While{$! \ \vactive.\mathsf{empty}()$}{
        \lIf{$\vhi.\mathsf{size}() + \vci.\mathsf{size}() > \vmaxInsts$}{\Break}
        $(\vtype, \vfront) := \vactive.\mathsf{front}()$ \;
        $\vactive.\mathsf{popFront}()$ \;
        \eIf{$\vtype = 0$}{
          $\vprevci := \vci.\mathsf{copy}()$ \;
          \For{$c : \vprevci$}{
            $\mathsf{matchOnePair}(c, \vfront, \vci, \vhi, \vactive)$
          }
        }{
          $\vprevhi := \vhi.\mathsf{copy}()$ \;
          \For{$h : \vprevhi$}{
            $\mathsf{matchOnePair}(\vfront, h, \vci, \vhi, \vactive)$
          }
        }
      }
      $\vmonohi := \mathsf{List.empty}()$ \;
      \For{$h : \vhi$}{
        \lIf{$\mathsf{QMono}(\Gamma; \emptyset, h)$}{$\vmonohi.\mathsf{push}(h)$}
      }
      \Return $\vmonohi$
    }
    \;
    \Fn{\matchOnePairFun{$c, h, \vci, \vhi, \vactive$}}{
      $\vnewhi := \mathsf{matchInst}(\Gamma; c, h)$ \;
      \For{$\vnh : \vnewhi$}{
        \lIf{$\vnh \in \vhi$}{\Continue}
        $\vhi.\mathsf{push}(\vnh); \vactive.\mathsf{push}((0, \vnh))$ \;
        $\vnewci := \mathsf{holInsts}(\Gamma; \emptyset, \vnh)$ \;
        \For{$\vnc : \vnewci$}{
          \lIf{$\vnc \in \vci$}{\Continue}
          $\vci.\mathsf{push}(\vnc); \vactive.\mathsf{push}((1, \vnc))$ \;
        }
      }
    }
  \end{algorithm}
  
  The saturation loop of quantifier instantiation is shown in Algorithm \ref{saturate}.\footnote{This
  is the same as Algorithm \ref{saturateInMain} in Sect. \ref{sectinst}.}
  For simplicity, equational theorem generation is not shown here.
  Given a $\lambda C$ context $\Gamma$ and a list $H$ of hypotheses, $\mathsf{saturate}$ returns
  a list of instances of hypotheses in $H$ that are suitable for $\lambda_\to^*$ abstraction
  (i.e. satisfy the $\mathsf{QMono}$ predicate).
  Note that, in Lean-auto, when checking whether a hypothesis instance belongs
  to a collection (e.g., set, list, queue, etc.) of hypothesis instances,
  we test equality only up to \textit{hypothesis equivalence}.
  
  \begin{definition}
  For two $\lambda C$ terms $t_1, t_2$, $t_1$ and $t_2$ are equivalent as hypotheses
  iff $t_1$ is a hypothesis instance of $t_2$ and $t_2$ is a hypothesis instance of $t_1$.\footnote{In
  higher-order logic and beyond, there exists terms that are instances of each other but not definitionally equal.}
  \end{definition}
  
  Checking membership up to equivalence ensures that collections of hypothesis
  instances in our algorithms are free of redundant entries. Note that equivalence testing
  can be reduced to unification, which can in turn be approximated by \texttt{isDefEq}.
\fi

\maybeappendix{Experiment on Translation}\label{sstrans}
\ifdefined\cameraReady
\else
  In this appendix, we present the result of our small-scale experiment on the comparison
  between encoding-based translation and monomorphization. We would like to compare the
  output sizes of the translation procedures on the same Lean 4 problem. For
  monomorphization, we use Lean-auto's translation procedure and compute the sum of the
  sizes of the output $\text{HOL}$ problem. Since Lean-auto does not support
  encoding-based translation, we use the size of the original Lean 4
  expression as the surrogate for the output size. This is justified by
  the fact that encoding-based translations usually produce outputs that are larger
  than the input problems.

  We randomly sample 512 user-declared theorems from Mathlib4. For each theorem, we generate its corresponding
  problem, which consists of the statement of the theorem and the statements of all the theorems
  used in its proof. The size of a problem is the sum of the sizes of all the expressions in
  the problem. We use Lean 4's deterministic timeout mechanism and set ``maxHeartbeats''
  to ``65536'' for the monomorphization of each problem, without imposing extra time or memory limit.

  Note that Lean-auto's monomorphization is incomplete, and it might be unfair to
  compare monomorphization with encoding-based translation on problems where
  monomorphization fails to produce a provable output. Therefore, we conduct another
  experiment with the Lean-auto-provable\footnote{Here we use Duper
  as the backend solver, and employ \textit{Experimental Setup 1}
  described in \refappendix{desetup}. The option ``auto.mono.ignoreNonQuasiHigherOrder'' is set
  to ``true'', and ``maxHeartbeats'' is set to ``65536''.} subset of the 512 problems. Note
  that if a problem is proved by Lean-auto, Lean-auto's monomorphization
  must have produced a provable output on the problem, regardless of the backend solver.

  \begin{figure}
    \begin{center}\begin{tabular}{| c | K{6em} K{6em} |}
      \hline
                        & Full        & Filtered   \\ \hline
      \#Theorems        & 512         & 188        \\
      \#Fails           & 88          & 0          \\
      Avg enc size      & 1503.4      & 643.5      \\
      Avg mono size     & 112.3       & 62.6       \\
      Avg (mono size)/(enc size) & 0.2325 & 0.2308 \\ \hline
    \end{tabular}\end{center}
    \caption{Result of experiment on translation} \label{figtrans}
  \end{figure}

  The result is presented in Figure \ref{figtrans}. ``\#Fails'' is the number of theorems
  where Lean-auto's monomorphization produces error. Failed theorems are not included
  when computing statistics. ``Avg enc size'' is the average size of the output of encoding-based
  translation. As mentioned before, we use the size of the original problem as an under-approximation.
  ``Avg mono size'' is the average size of the monomorphized problem. ``Avg (mono size)/(enc size)''
  is the average ratio of the monomorphized size and the encoding-based size. The result
  indicates that monomorphization produces significantly smaller results compared to encoding-based
  translation.
\fi

\maybeappendix{Experiment on Reduction}\label{ssred}
\ifdefined\cameraReady
\else
  In this appendix, we investigate the possibility of reducing the input expressions
  before sending them to Lean-auto. When reducing expressions, Lean 4 allows users to
  control which constants are unfolded, with three \textit{transparency levels}: \textit{reducible},
  \textit{default} and \textit{all}. In the \textit{reducible} level, only a small portion of
  constants are unfolded. Lean-auto reduces all input expressions with the \textit{reducible} level,
  because this helps alleviate the definitional equality problem,
  and usually don't increase the expression size by too much. In the \textit{default} level,
  most non-theorem constants are reduced. Reducing with \textit{default} level will make
  many definitionally equal input expressions become syntactically identical, but might
  make the expressions become unacceptably large. In the \textit{all} level, all constants are
  unfolded (except for those marked with the special tag \textit{opaque}). Reducing with
  \textit{all} level will produce even larger expressions than with the \textit{default} level.

  We use the same 512 Mathlib4 theorems in \refappendix{sstrans}, and generate their corresponding
  problems in the same way. Experiment is conducted on Amazon EC2 \texttt{c5ad.16xlarge}. The
  time limit for each problem is 120 seconds, and the memory limit is 8GB.

  \begin{figure}
    \begin{center}\begin{tabular}{| c | K{6em} K{6em} K{6em} |}
      \hline
                        & reducible   & default          & all               \\ \hline
      \#Fails           & 0           & 83               & 202               \\
      Avg size before   & 791.5       & 588.3            & 487.7             \\
      Avg size after    & 2449.8      & 138579513.0      & 258118331.0       \\
      Avg size increase & 5.8$\times$ & 309146.5$\times$ & 1216555.0$\times$ \\
      \#$10 \times$ increase &    48  &              215 & 151               \\
      \#$10 \times$ increase + \#Fails & 48 &        298 & 353               \\ \hline
    \end{tabular}\end{center}
    \caption{Result of experiment on reducing input expressions} \label{figredinput}
  \end{figure}

  The result is presented in Figure \ref{figredinput}. ``\#Fails'' is the number
  of problems that exceeds time or memory limit. This represents the problems
  which are complex enough such that running reduction on them are prohibitively expensive.
  For each transparency level, the problems it fails on are excluded when computing its statistics. ``\#$10\times$ increase''
  is the number of problems whose size increases to at least $10\times$ its original
  size after reduction. Therefore, ``\#$10 \times$ increase + \#Fails'' roughly corresponds
  to the problems that become much harder to prove after reducing with the given
  transparency level. According to ``\#$10 \times$ increase + \#Fails'', both the \textit{default} and \textit{all} level
  produce unacceptable results on at least 50\% of the theorems, while for \textit{reducible}
  it's less than 10\%. This suggests that we should not reduce the input problem with
  \textit{default} or \textit{all} level, and therefore should handle the definitional equality problem
  using other methods.
\fi

\maybeappendix{Experiment on Duper}\label{ssduper}
\ifdefined\cameraReady
\else
  We conduct a small-scale experiment to compare the performance of Duper with
  and without Lean-auto. We use the same 512 Mathlib4 theorems in \refappendix{sstrans}, and generate their
  corresponding problems in the same way. We use Lean 4's deterministic timeout mechanism
  for resource control and set the timeout option ``maxHeartbeats'' to 65536.
  The option ``auto.mono.ignoreNonQuasiHigherOrder'' of Lean-auto is set to ``true''.
  As explained in \refappendix{desetup}, we employ \textit{Experimental Setup 1}
  in this experiment.

  \begin{figure}
    \begin{center}\begin{tabular}{| l | K{6em} K{7em}|}
      \hline
                        & Solved      & Avg Time(ms) \\ \hline
      With Lean-auto    & 189(36.9\%) & 1375.7       \\
      Without Lean-auto & 42(8.2\%)   & 1856.3       \\ \hline
    \end{tabular}\end{center}
    \caption{Comparison of Duper with and without Lean-auto.} \label{figssduper}
  \end{figure}

  We see that when Duper is used without Lean-auto, it only solves 8.2\% of the
  problems, and it is slower on solved problems compared to ``Duper with Lean-auto''.
  Duper also exhibits unexpected behaviors during the experiment.
  We find that Duper gets stuck on 7 of the 512 problems
  for more than 5 minutes. Moreover, we find that Duper spends 1741174 heartbeats
  on the theorem ``Measure\-Theory.Lp.simple\-Func.is\-Dense\-Em\-bedding'' before failing, which
  vastly exceeds our limit 65536. We suspect that in these cases, Duper runs into
  code not controlled by Lean 4's deterministic timeout mechanism.
  
  When we attempted full-scale evaluation of ``Duper without Lean-auto''
  on Mathlib4, we found similar issues. Duper gets stuck on problems for
  minutes and even hours. Manually recording these problems and filtering
  them out would require significant manual work.\footnote{Similar issues
  are also present when evaluating other tools, but are much
  less pronounced compared to ``Duper without Lean-auto''. Therefore, we
  were able to manually filter out these problems.} Therefore, we decided to
  not include ``Duper without Lean-auto'' in our full-scale evaluation.
\fi

\maybeappendix{Details on Theorem Proving Experiments}\label{desetup}
\ifdefined\cameraReady
\else
  Multiple experiments in this paper involve running Lean-auto
  or existing tools on Lean 4 theorems. Here, we present technical
  details of experimental setups used in these experiments.
  
  All the tools we evaluate, including Lean-auto and existing tools, are implemented
  as tactics in Lean 4. Each tactic in Lean 4 has a user-facing syntax and an underlying
  tactic function. To invoke a tactic, users can input the syntax of the tactic in Lean 4,
  potentially with extra information (such as a list of premises). Lean 4 will elaborate
  the syntax and call the underlying tactic function.

  A straightforward way to evaluate a tactic \texttt{tac} on a list $\mathit{Ts}$ of Mathlib4
  theorems is shown as \textit{Experimental Setup 1} in Figure \ref{figexs1}.

  \begin{figure}
    \begin{center}
    \fbox{\parbox{30em}{
      To run a tactic \texttt{tac} on a list $\mathit{Ts}$ of Mathlib4 theorems:

      \begin{enumerate}
        \item Import the entire Mathlib4
        \item For each theorem $T$ in $\mathit{Ts}$, collect all the theorems $h_1, \dots, h_n$ used
          in the proof of $T$. Then, call the underlying tactic function
          of \texttt{tac} on the statement of $T$ and record the result. If \texttt{tac} accepts premises, supply
          $h_1, \dots, h_n$ as the list of premises to the underlying tactic function.
      \end{enumerate}
    }}
    \end{center}
    \caption{Experimental Setup 1}\label{figexs1}
  \end{figure}

  However, \textit{Experimental Setup 1} is unfair because it favors \texttt{simp\_all}
  and \texttt{aesop}. This is related to the fact that these two tactics have access to
  theorems tagged with the ``simp'' attribute. Suppose a theorem $T$ in Mathlib4 is tagged with ``simp''.
  If we run \texttt{simp\_all} on $T$ after importing Mathlib4, then \texttt{simp\_all}
  will have access to the ``simp''-tagged $T$, which might cause it
  to find a proof of $T$ that uses $T$ itself.

  Therefore, we would like to make sure that a theorem $T$ is not already tagged with ``simp''
  when we run evaluation on $T$. A way to achieve this is to retrieve the Lean 4 file that declares
  $T$, execute all the commands \textit{before} the declaration of $T$, then run evaluation on
  the statement of $T$. This makes sure that $T$ is not declared (thus not marked with
  ``simp'') when we run evaluation on it.

  However, this method causes another problem. There are commands in Lean 4 that simultaneously
  decare multiple constants $c_1, \dots, c_n$. If there exists $i, j$ such that
  $c_i$ is a theorem and $c_j$ occurs in the statement of $c_i$, then running evaluation
  using the above method on $c_i$ will cause an ``unknown constant'' error,
  because $c_j$ is not declared when we run evaluation on the statement of $c_i$.
  Similarly, if $c_j$ occurs in the proof of $c_i$, then running evaluation using the above
  method is also problematic because the not-yet-declared $c_j$ would be passed to
  those tools that accept premises. Therefore, we would like to filter out theorems
  whose proof or type contains constants declared by the same command.

  To make our evaluation more closely resemble real use cases of Lean 4, we
  would like to invoke the user-facing syntax of the tactics instead of their underlying functions.
  This causes some more fails for premise-accepting tactics because
  many Mathlib4 proofs use non-user-declared theorems that are inaccessible to users.

  Our modified evaluation method is presented as \textit{Experimental Setup 2}
  in Figure \ref{figexs2}. We employ a per-file evaluation scheme
  for better efficiency.

  \begin{figure}
    \begin{center}
    \fbox{\parbox{30em}{
      To run a tactic \texttt{tac} on a Mathlib4 file $F$:

      \begin{enumerate}
        \item Retrieve the content of $F$
        \item For each command $C$ in $F$:
          \begin{enumerate}
            \item Record the environment $E$ before executing $C$. $E$
              contains all the constants declared by commands prior to $C$.
            \item Run command $C$ and record the constants $c_1, \dots, c_n$
              declared by it.
            \item Record the environment $E'$.
            \item Set the environment to $E$. This effectively removes
              $c_1, \dots, c_n$ from the environment.
            \item For each $1 \leq i \leq n$, if $c_i$ is a theorem and does
              not contain $c_j (1 \leq j \leq n)$ in its proof or type:
              \begin{enumerate}
                \item Collect all the theorems $h_1, \dots, h_n$ used in the proof of $c_i$
                \item Create the syntax $S$ that invokes \texttt{tac} on $c_i$. If \texttt{tac}
                  accepts premises, $h_1, \dots, h_n$ should be supplied to \texttt{tac} in the syntax.
                \item Run Lean 4 on $S$ and record the result.
              \end{enumerate}
            \item Set environment to $E'$. This adds back constants declared by $C$,
              which is necessary to the execution of later commands.
          \end{enumerate}
      \end{enumerate}
    }}
    \end{center}
    \caption{Experimental Setup 2}\label{figexs2}
  \end{figure}

  The experiments in Sect. \ref{sectexpr} employ \textit{Experimental Setup 2}.
  For other small-scale experiments in our paper, we use \textit{Experimental Setup 1}.
  This is because these small-scale experiments do not involve \texttt{simp\_all} and \texttt{aesop},
  and \textit{Experimental Setup 1} is a cleaner evaluation method compared to \textit{Experimental Setup 2}.

  Now, we discuss details of resource limit and benchmark generation.

  \subsubsection{Resource Limit:} For efficiency reasons, we would like each
  Lean 4 process to test multiple problems (instead of one problem per process).
  Lean 4 does not support setting time limit or memory limit for native code.
  Instead, it provides a resource control mechanism called \textit{deterministic timeout},
  which is controlled by the ``maxHeartbeats'' option. The deterministic timeout
  mechanism counts the number of times a low-level Lean 4 function is called, and interrupts
  the program if it exceeds ``maxHeartbeats''.
  
  In the experiments in Sect. \ref{sectexpr}, we mentioned that all the tools are
  given a time limit of 10 seconds. For native Lean 4 tools, including ``rfl'',
  ``simp\_all'', ``Aesop'' and ``Lean-auto + Duper'', we set ``maxHeartbeats''
  to 65536, which we have found to roughly correspond to 10 seconds in our experiments.
  For ``Lean-auto + TPTP/SMT Solver'', we set ``maxHeartbeats'' to 65536
  for Lean-auto's native Lean 4 code, and set timeout to 10 seconds for TPTP and SMT Solvers.
  Note that the setups are not imposing a strict 10 seconds limit on any of the tools.
  Therefore, we also record the total execution time of each tool on each problem,
  and problems that takes more than 10 seconds to solve are counted as fails.

  Note that the above discussion only applies to Sect. \ref{sectexpr}.
  For other small-scale experiments in our paper, since they do not involve
  external solvers, we set ``maxHeartbeats'' to a fixed value without imposing extra
  time or memory limits.

  \subsubsection{Benchmark Generation:} Either of \textit{Experimental Setup 1}
  or \textit{Experimental Setup 2} naturally gives rise to a benchmark generation method.
  For \textit{Experimental Setup 1}, the corresponding benchmark set is all the
  user-declared Mathlib4 theorems\footnote{A constant is a Mathlib4 constant
  iff it is declared by a \texttt{.lean} file in Mathlib4. Note that the environment
  after importing Mathlib4 also contains constants declared in libraries that Mathlib4
  depend on.} in the environment after importing Mathlib4, which amounts to 178026 theorems. For
  \textit{Experimental Setup 2}, the corresponding benchmark set is all the
  user-declared theorems generated by the commands (in the Mathlib4 files) executed
  during the experiment. We find a slight difference (around 100 theorems) in the benchmark
  sets generated by \textit{Experimental Setup 2} when testing different tools.
  This is potentially due to issues related to individual tools.\footnote{
  Note that in \textit{Experimental Setup 2}, execution of tools interleave with
  execution of commands in Mathlib4 files, and execution of commands produce
  constants, which are then filtered to produce the benchmark set. If the tool crashes or
  causes other side effects, it could affect the constants produced by the commands.}

  In the experiments in Sect. \ref{sectexpr}, the benchmark set we use is the
  intersection of the above two benchmark sets. For each Mathlib4 file $F$, we record
  both the set of theorems from $F$ after importing the entire Mathlib4 and the set of theorems generated by
  executing commands in $F$, then compute the intersection of the two sets. This gives a total of
  176904 theorems. After filtering out the 27762 theorems whose proof or type contains
  constants declared in the same command, our final benchmark set consists of 149142 theorems.

  Note that the above benchmark generation method only applies to Sect. \ref{sectexpr}.
  For our small-scale experiment, we randomly sample from user-declared Mathlib4 theorems
  in the environment after importing Mathlib4.
\fi

\end{document}